\documentclass[aps,showpacs, preprint, 12pt,,superscriptaddress,superscriptemail]{revtex4}
\pdfoutput=1
\usepackage{graphics}
\usepackage{graphicx}
\usepackage{amsmath}
\usepackage{amsfonts}
\usepackage{amssymb}
\usepackage{bm}
\usepackage{bbm}

\newtheorem{theorem}{Theorem}[section]

\newtheorem{definition}[theorem]{Definition}

\newtheorem{lemma}[theorem]{Lemma}

\newtheorem{proposition}[theorem]{Proposition}

\newenvironment{proof}[1][Proof]{\noindent\textbf{#1.} }{\ \rule{0.5em}{0.5em}}
\numberwithin{equation}{section}

\begin{document}
\title{Universal exit probabilities in the TASEP}

\author{S.S. Poghosyan}
\email{spoghos@theor.jinr.ru}
\affiliation{Bogoliubov Laboratory of Theoretical Physics, Joint Institute for Nuclear Research, 141980 Dubna, Russia}
\author{A.M. Povolotsky}
\email{alexander.povolotsky@gmail.com}
\affiliation{Bogoliubov Laboratory of Theoretical Physics, Joint Institute for Nuclear Research, 141980 Dubna, Russia}
\affiliation{National Research University Higher School of Economics, 20 Myasnitskaya Ulitsa, Moscow 101000, Russia }
\author{V.B. Priezzhev}
\email{priezzvb@theor.jinr.ru}
\affiliation{Bogoliubov Laboratory of Theoretical Physics, Joint Institute for Nuclear Research, 141980 Dubna, Russia}

\bigskip
\begin{abstract}We study the joint exit probabilities of particles in the totally asymmetric simple exclusion process (TASEP)
from space-time sets of given form. We extend previous results on the space-time correlation functions of the TASEP, which correspond to exits from the sets bounded by straight vertical or horizontal lines. In particular, our approach allows us to remove ordering of time moments used in previous studies so that only a natural space-like ordering of particle coordinates remains. We consider sequences of general staircase-like boundaries going from the northeast to southwest in the space-time plane. The exit probabilities from the given sets are derived in the form of Fredholm determinant defined on the boundaries of the sets. In the scaling limit, the staircase-like boundaries are treated as approximations of continuous differentiable curves. The  exit probabilities with respect to  points of these curves   belonging to arbitrary space-like path are shown to converge to the universal Airy$_2$ process.
\end{abstract}
\pacs{05.40.+j, 02.50.-r, 82.20.-w}
\maketitle
\section{Introduction}

Consider the system of particles on the 1D integer lattice. At any time
moment a configuration of particles is specified by a set of $N$ strictly
increasing integers, $(x_1>x_2>\dots)$, denoting particle coordinates.
They evolve in a discrete time $t\in\mathbbm{Z}$ according to the TASEP \cite{Liggett}
dynamical rules:

\begin{enumerate}
\item[I.] A particle takes a step forward, $(x_i\to x_{i}+1)$, with
probability $p$ and stays at the same site, $(x_i\to x_i)$, with probability
$q\equiv1-p$ provided that the target site is empty, $(x_i+1\neq x_{i-1})$.

\item[II.] If the next site is occupied, $(x_i+1=x_{i-1})$, the particle
stays with probability $1$.

\item[III.] The backward sequential update is used \cite{update}: at each time step the
positions $x_i$ of all particles are updated one by one, in the order of
increasing of particle index: $i=1,2,3,\dots$
\end{enumerate}

These dynamical rules define transition probabilities for a Markov chain
constructed on the set of particle configurations. Given initial conditions,
one can inquire for probabilities of different events in course of
the Markov evolution. In present paper, we are interested in the correlation
functions which are the probabilities for events associated with a few specified
particles and given space-time positions.

\subsection{Spacial correlation functions of the TASEP.}

The first exact result on correlation functions in TASEP goes back to
prominent Johansson's work \cite{Johansson}, where he considered the
evolution of TASEP with parallel update and step initial conditions,
\begin{equation}  \label{step}
\bm{x}^0=\mathbbm{Z}_{\leq 0},
\end{equation}
and obtained the distribution, $P_t(x_N>M-N)$, of the distance $M$ traveled by $N
$-th particle up to time $t$. This result was later generalized to the
backward sequential update \cite{RS} and the flat initial conditions \cite{Nagao}.
The connection of the TASEP with the theory of determinantal point
processes revealed in \cite{Sasamoto,BFPS} allowed also calculation of the
multi-particle correlation functions, i.e. distribution $P_t(x_{n_1}>a_1,%
\dots x_{n_m}>a_m)$, of positions of $m$ selected particles at fixed time $t$%
, where $1\leq n_1< \dots<n_m $ are $m$ integers numbering the selected
particles. The multi-particle correlation functions were extensively studied
for different initial conditions in a series of papers \cite{BFPS,BFP,BFS1}. The
result can generally be represented in a form of the Fredholm determinant of
the operator with some integral kernel. An asymptotic analysis of the kernel
is of special interest as it allows one to study the scaling limit of the
correlation functions, which is believed to yield universal scaling
functions of the Kardar-Parisi-Zhang (KPZ) universality class \cite{KPZ}.

There is a law of large numbers, which implies that the stochastic evolution
converges to a deterministic limit \cite{Rost,Spohn}. Specifically, in the TASEP,
if we measure coordinate $x_n$ of $n$-th particle at time $t$, the
deterministic relation between rescaled variables
\begin{equation}  \label{scal lim vars}
\nu\equiv n/L,\omega \equiv t/L,,\gamma \equiv (x_n+n)/L
\end{equation}
holds with probability one as $L\to\infty$. An explicit form of this
relation can be found from the hydrodynamic conservation law
\begin{equation}  \label{hydrodynamics}
\partial_t \rho+\partial_x j=0.
\end{equation}
for the density of particles $\rho$. Here $j\equiv j(\rho)$ is the
stationary current of particles, which is a model-dependent function of the
density. In the case of backward update the current is
\begin{equation}  \label{current}
j(\rho)=\frac{p \rho(1-\rho)}{1-p\rho}.
\end{equation}
Then, the solution of   (\ref{hydrodynamics}) with initial conditions (\ref{step}), yields relation
\begin{equation}  \label{det lim rel}
\sqrt{p\omega}-\sqrt{q\nu}-\sqrt{\gamma}=0,
\end{equation}
which holds in the range $- p/q \leq (\gamma-\nu)/\omega \leq p$. For the the formula \eqref{current} and its relation to
\eqref{det lim rel} we address the reader to references \cite{Johansson,RS}.

An exact calculation of the correlation functions allows one to study
fluctuations of the random variables near their value on the deterministic
scale. Given $\nu$ and $\omega$, let $\gamma(\omega,\nu)$ be the rescaled
particle coordinate. The deviation $\delta x_n\equiv x_n-L(\gamma(\omega,\nu)-\nu)$
of the particle coordinate $x_n$ develops on the KPZ characteristic scale
fluctuations
\begin{equation}  \label{KPZ scaling}
\delta x_n \sim L^{\alpha},\alpha=1/3.
\end{equation}
The distribution of the rescaled variable
\begin{equation}  \label{s}
s= \kappa_x^{-1} \lim_{L\to\infty} \delta x_n L^{-\alpha},
\end{equation}
is a universal scaling function of the KPZ class, dependent only on the form
of the initial macroscopic density profile. Note that the model dependence
is incorporated into a single non-universal constant $\kappa_x$. The
examples of distributions obtained from the asymptotic analysis of the
one-point correlation function are the Tracy-Widom functions $F_1$ and
$F_2$ for flat and step initial conditions respectively. These functions are
well known for appearing in the theory of random matrices as the
distributions of the largest eigenvalue in the orthogonal and unitary
Gaussian ensembles \cite{TW,Mehta}. Their presence turns out to be a universal
feature of the KPZ class. Furthermore, the study of multipoint distributions
shows that the fluctuations of  coordinates of different particles, say $%
x_{n_1}$ and $x_{n_2}$, remain non-trivially correlated random variables on
the scale
\begin{equation}  \label{beta}
|n_1-n_2|\sim L^{\beta}, \beta=2/3.
\end{equation}
This is the second power law characterizing the KPZ class. The critical
exponents $\alpha=1/3$ and $\beta=2/3$ are called fluctuation and
correlation exponents respectively. After corresponding rescaling of
particle numbers, one arrives at the one-parametric family of correlated
random variables:
\begin{equation}  \label{s(u)}
s(u)=\lim_{L\to\infty}\frac{x_{[\nu N+u N^{\beta} \kappa_n
]}-L(\gamma(\omega, \nu +u L^{\beta-1} \kappa_n)-\nu-u L^{\beta-1} \kappa_n)}{L^\alpha \kappa_x},
\end{equation}
where $\kappa_n$ is another non-universal constant. For the cases of flat
and step initial conditions, the joint distributions of these variables
define universal Airy$_1$ \cite{Sasamoto} and Airy$_2$ \cite{PraehoferSpohn}
ensembles, whose one-point distributions are $F_1$ and $F_2$.

\subsection{Space-Time correlations and mapping to the last passage
percolation.}

So far we have been discussing only the spacial correlations between
positions of different particles at a fixed time moment. However, generally,
one can consider joint probability distributions of events associated with
different particles, positions and time moments, which happen in course of
the TASEP evolution. We will refer to these distributions as the space-time
correlation functions. An example of such a function, the distribution of
positions of a tagged particle at different moments of time, has been calculated
in \cite{ImamSasamoto}. A more general correlation function, the
distribution $P(x_{n_1}(t_1)>a_1,\dots,x_{n_m}(t_m)>a_m)$ of positions $%
x_{n_1},\dots,x_{n_m}$ of selected particles with numbers
\begin{equation}  \label{n order}
n_1\leq\dots\leq n_m
\end{equation}
at time moments $t_1,\dots,t_m$, was studied in \cite{BorodinFerrari,BFS}.
The method was used that restricted the analysis to the sets of space-time
points, such that the time coordinates decreased weakly with the particle
number and vice versa:
\begin{eqnarray}  \label{space-like 1a}
t_i&\geq& t_{i+1}, \,\,\mathrm{if} \,\, n_i<n_{i+1}, \\
t_i&>&t_{i+1},\,\, \mathrm{if} \,\, n_i=n_{i+1} .  \label{space-like 1b}
\end{eqnarray}
This arrangement of time moments was named space-like by the authors of  \cite{BorodinFerrari,BFS}. Another example of the space-time correlation
function, the current correlation function, was recently obtained in \cite{PovPrS}. This was the probability distribution $P(t_{n_1}<a_1,%
\dots,t_{n_m}<a_m)$ of time moments $t_{n_1},\dots,t_{n_m}$ at which $m$
selected particles with numbers
\begin{equation}  \label{order_n 2}
n_1<\dots<n_m
\end{equation}
jump from the respective sites $x_{n_1},\dots,x_{n_m}$ selected from the set
\begin{equation}  \label{fence}
\{x_i=x-i+N:\, i=1,\dots, n_m\},
\end{equation}
given $x\in \mathbbm{Z}$, $N\geq n_m$ and the initial configuration $%
x_{i}^0=1-i, i\in \mathbb{N}$. Due to non-crossing of space-time particle
trajectories, the range of time moments accessible for the dynamics is
\begin{equation}  \label{time order current}
t_{n_1}\leq\dots\leq t_{n_m}.
\end{equation}

The time orderings (\ref{space-like 1a},\ref{space-like 1b}) and (\ref%
{order_n 2},\ref{time order current}) are opposite to each other. These
orderings, however, have different origins. In \cite{ImamSasamoto,BorodinFerrari,BFS}, numbers of particles $n_1,n_2,\dots$ and
time moments $t_{n_1},t_{n_2},\dots$ are fixed, and particle coordinates $%
x_{n_1},x_{n_2},\dots$ are random variables. In the case of current
correlations \cite{PovPrS}, time moments $t_{n_i}$ are random, while
particle coordinates $x_{n_i}$ and numbers $n_i$ are related fixed
parameters. Therefore, unlike (\ref{space-like 1a},\ref{space-like 1b}) in
\cite{BorodinFerrari,BFS}, (\ref{time order current}) from \cite{PovPrS} is
not an external constraint, but is the consequence of dynamics: it shows
domains which can be reached in the random process with nonzero probability.

Which variable is chosen to be random is, however, not important in the
scaling limit, when the three variables, time and space coordinate and the
number of a particle, acquire   equivalent significance due to separation of fluctuation and
correlation scales. Indeed, once we have fixed the values of any two of the
parameters $n,x,t$ on the large scale, the value of the third one is
uniquely fixed to the same order by the deterministic relation (\ref{det lim
rel}).  Then, the random fluctuations of any of these quantities characterize the degree
of violation of this relation.
In other words, we fix a point on the 2D surface defined by the
relation (\ref{det lim rel}) in 3D space of parameters $\gamma,\omega,\nu$. Then, the small fluctuations in
the vicinity of this point are represented by an infinitesimal vector
normal to the surface, which can be projected to one of three directions $\gamma,\omega,\nu$
 or any other direction in 3D space. A choice of the direction affects
only the angle-dependent constants defining the fluctuation scale, while the
functional form of the distributions is universal. Furthermore, the
correlations between fluctuations associated with different points of the
surface are also universal, as far as the points are separated by a distance
of order of correlation scale, $N^\beta$. The universality holds as the
mutual positions of the points vary in a wide range. Indeed, the limiting
correlation functions of both positions \cite{ImamSasamoto,BorodinFerrari,BFS} and times \cite{PovPrS} chosen within the
domains (\ref{space-like 1a},\ref{space-like 1b}) and (\ref{order_n 2},\ref%
{time order current}), respectively, yield $Airy_2$ correlations for the
case of step initial conditions.

How rigid the universality with respect to the choice of points within the
correlation function was clarified by Ferrari in \cite{Ferrari}, whose
arguments were based on the observed slow decorrelation phenomena. He
explained that the limiting correlations can be of two types depending on
whether the point configurations under consideration are space-like or
time-like. The correlations for the space-like configurations are, up to a
non-universal scaling factor, of the same form as the purely spacial
correlations. Specifically, when the distance between points is of order $%
N^\beta$, the fluctuations at these points are described by the Airy$_1$,
Airy$_2$ e.t.c. ensembles, depending on the initial conditions, like in the
purely spacial case. However, if the point configuration is time-like, the
fluctuations, measured at the characteristic fluctuation scale $N^\alpha$,
remain fully correlated, i.e. identical, until the distance between the
points will be of order of $N$, which is much larger than $N^\beta$.

The definitions of space-like and time-like point configurations used in
\cite{Ferrari} for the polynuclear growth (PNG) model and extended by
Corwin, Ferrari and Peche (CFP), \cite{CorwinFerrariPeche}, to a wide range
of other models including TASEP were, however, different from the one
accepted in \cite{BorodinFerrari,BFS}. To classify  our results correctly,
we recap here the main idea of CFP. Their formulation used the language of
the last passage percolation \cite{Johansson}, which can be directly, mapped
to the TASEP as well as to many different models \cite{CorwinFerrariPeche}.
Let $\mathbbm{R}^2_+$ be the first quadrant of $\mathbbm{R}^2$. Each point
of $\mathbbm{R}^2_+$ with positive integer coordinates $(i,j)\in\mathbbm{N}%
^2\bigcap\mathbbm{R}^2_+$, is assigned a geometrically distributed random
variable $T_{i,j}$,
\begin{equation}  \label{geometr}
P(T_{i,j}=t)=q^t (1-q).
\end{equation}
A particular realization of the TASEP evolution is recorded in the values of
$T_{i,j}$. Namely, $T_{i,j}$ is the time the $i$-th particle is waiting for
before making $j$-th step after it has been allowed to move. A directed
lattice paths, $\Pi_{(x_1,y_1)\to(x_2,y_2)}$, is the path, which starts at
the point $(x_1,y_1)$ and, making only unit steps either upward, $%
(i,j)\to(i,j+1)$, or rightward, $(i,j)\to(i+1,j)$, ends at the point $%
(x_2,y_2)$. The sum of $T_{i,j}$ over the path is referred to as the last
passage time. As it was shown by Johansson for the TASEP with parallel
update \cite{Johansson}, the last passage time, maximized over the set of
all paths from $(1,1)$ to $(n,m)$,
\begin{equation}
\mathcal{T}_{n,m}= \max_{\{\Pi_{(0,0)\to(n,m)}\}} \sum_{(i,j)\in \Pi}T_{i,j}
\end{equation}
is related to time $t_n(m)$ the $n$-th particle takes to make $m$ steps, $%
t_n(m)=\mathcal{T}_{n,m}+n$. For the TASEP with backward sequential update
these two times are simply equal, $t_n(m)=\mathcal{T}_{n,m}$. Other models
can be obtained as limiting cases. In the limit $q\to 1$ with rescaling of
time $t\to t(1-q)$ we obtain the exponential distribution of waiting times,
which defines the continuous time TASEP. In the opposite limit $q\to 0$ the
first quadrant is filled mainly by zeroes, while ``one'' appears rarely
having concentration $q$. After going to the continuous limit with rescaled
coordinates $(x,y)\to (qx,qy)$, the distribution of ``ones'' on the
background of zeroes becomes the Poisson process in the first quadrant,
which in turn can be used to define the PNG \cite{PraehoferSpohn,BorodinOlshanski}. Given $(n,m)$, the probability
distribution of waiting times (\ref{geometr}) induces the distribution $P(%
\mathcal{T}_{n,m}<a)$ of the last passage time $\mathcal{T}_{n,m}$. The
joint distributions $P(\mathcal{T}_{n_1,m_1}<a_1,\dots,\mathcal{T}%
_{n_k,m_k}<a_k)$ of the last passage times for $k$ different points $%
(n_1,m_1),\dots,(n_k,m_k)$ are referred to as $k$-point correlation
functions.

According to CFP, two-point configuration $((n_1,m_1),(n_2,m_2))$ is
time-like if the points can be connected by a directed path $%
\Pi_{(n_1,m_1)\to (n_2,m_2)}$ and is space-like otherwise.  Suppose that $%
n_1 \leq n_2$. Obviously,  the time-like conditions are
\begin{eqnarray}  \label{space-like ferrari}
m_1 \leq m_2\,\,\mathrm{when}\,\, n_1<n_2, \\
m_1 <m_2\,\,\mathrm{when}\,\, n_1=n_2.  \notag
\end{eqnarray}
Recall that in the TASEP with step initial conditions a particle with the
number $n$ starts at initial position $x_n^0=-n+1$. Therefore, the spatial
coordinate of the particle, which has traveled for the distance $m$, is $%
x_n=m-n+1$. Then, the space-like condition opposite to (\ref{space-like
ferrari}) can be translated to the one for the space coordinates:
\begin{eqnarray}  \label{space-like ferrari coord}
x_i< x_j\,\,\mathrm{when}\,\, i\geq j.
\end{eqnarray}
This is the condition that the slow decorrelation does not occur, and,
correspondingly, the universality holds. One can see that the points (\ref%
{fence}) of final configurations within the current correlation functions
satisfy this condition. Also, due to non-crossing of particle trajectories,
these conditions hold automatically when the time moments are chosen in the
domain (\ref{space-like 1a},\ref{space-like 1b}). Therefore, the point
configurations studied in \cite{ImamSasamoto,BorodinFerrari,BFS,PovPrS} are
space-like according to CFP classification. However, in the complementary
domain, both types of the scaling behaviour present. Thus, the division to
time-like and space-like configurations proposed by CFP is more adequate if
one wants to distinguish between different types of universal behaviour of
correlation functions. By this reason, we keep on their terminology, where
the space-like configurations in TASEP are defined by the condition (\ref%
{space-like ferrari coord}) and time-like by the opposite one. The current
correlation functions calculated in \cite{PovPrS} were just an example of
space-like correlations beyond the domain studied in \cite%
{BFS,BorodinFerrari}. In fact, the earliest result on space-like
correlations was obtained in \cite{BorodinOlshanski}, where the universality
of the scaling limit was shown in context of the PNG model in the whole
space-like domain. However the microscopic consideration in context of the
TASEP was limited to (\ref{space-like 1a},\ref{space-like 1b})
in \cite{ImamSasamoto,BFS,BorodinFerrari,PovPrS} and to (\ref{order_n 2},\ref{time
order current}) in \cite{PovPrS}, where the spacial coordinates were fixed
by (\ref{fence}).

In this paper we extend the microscopic
derivation of the TASEP correlation functions to the rest of the space-like
domain, what has not been covered by previous analysis.

\subsection{General overview and the aim of the present work.}

We conclude the introductory part with an informal outline of the recent
development of the theory of multipoint correlation functions described above and formulation
of purposes we are going to fulfil below.

Though the previous  results were formulated in terms of distributions of various quantities, they can be considered
in a similar fashion if we look at the TASEP as at the probability measure over collections of interacting lattice paths (the space-time trajectories of particles), which can go one step down (particle stays) or down-right (particle makes a step) in the space-time plane.
Then the correlation functions give  marginal probabilities of certain points or bonds of the underlying lattice to belong to paths corresponding to selected particles. Specifically the development can be roughly divided into three stages depicted in Fig.(\ref{fig0}).
At the first stage the  points  were fixed at the same  moment
of time, e.g. those  encircled in Fig.(\ref{fig0}a).
\begin{figure}[tbp] \centering
\includegraphics[width=0.9\textwidth]{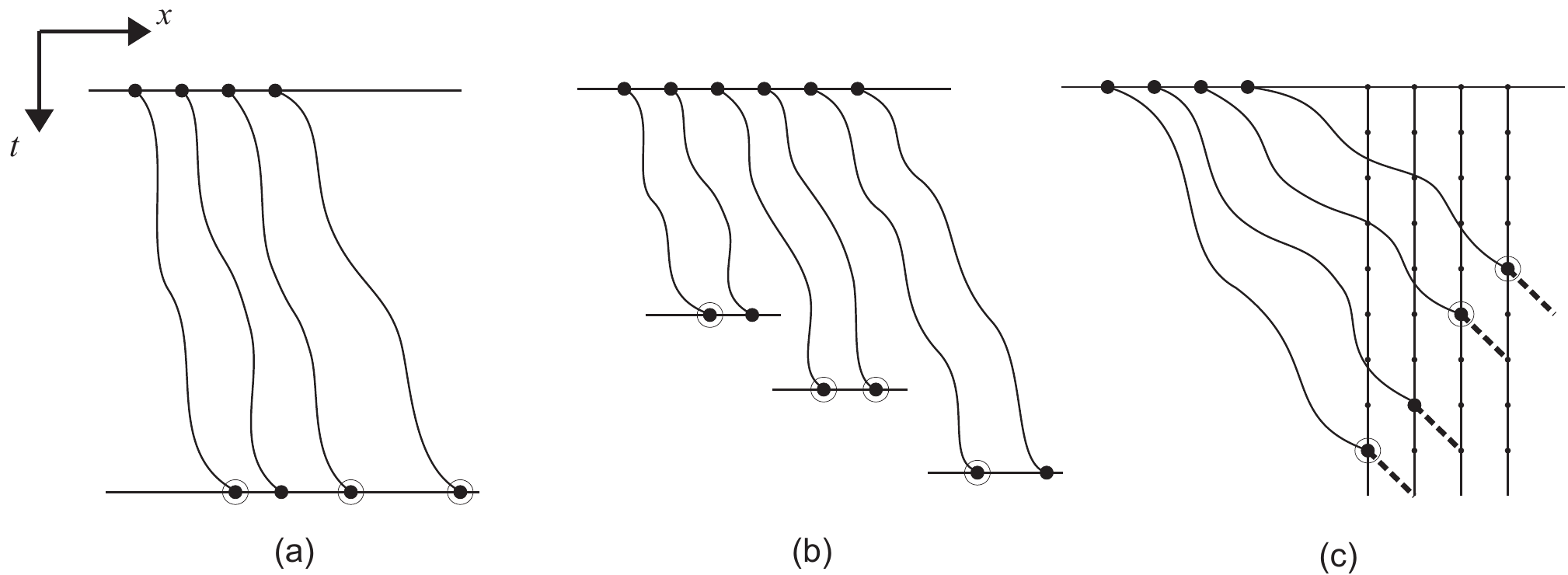}
\caption{(a) Equal-time correlation functions; (b) The first extension of the space-time domain;
(c)Current correlation functions.}
\label{fig0}
\end{figure}
The basic achievement of this stage, mentioned in subsection 1.1,
is revealing the structure of determinantal process in the TASEP \cite{Sasamoto,BFPS}.

The second stage described in subsection 1.2 is characterized by an extension of the range of point configurations to space-time domain shown in Fig.(\ref{fig0}(b)). The  condition crucial for the solution is the possibility to cut off
the part of particle trajectory following the selected point without affecting the remaining part. In the first case we just stopped at the moment of interest and the independence from the future  was a trivial consequence of the fact that the TASEP is a Markov process. In the second case similar independence  follows from another Markov property specific for  the TASEP dynamics \cite{BorodinFerrari,BFS}:
the particles in the TASEP do not affect an evolution of other particles to the right of them. Therefore, one can drop a part of a particle trajectory if there is no points fixed to the left of it at later time, see Fig.(\ref{fig0}(b)), so that the time corresponding to the selected points increases weakly from left to right. Finally one again arrives at the determinatal process, though more elaborated than the one in the first case.

The third stage, referred to as  current correlation functions, is depicted in Fig.(\ref{fig0}c).
Here the particle trajectories propagate equal distances in spatial direction
and the selected points are fixed at different moments of time, which, as seen from the picture,  must increase weakly from right to left.
At the first glance this situation is in contradiction with the above "trajectory cutting" ideology. However it is not difficult to
convince oneself that if we require that the trajectory makes a step forward after the selected point, it has no chance to interact with the trajectory
that ends one step to the left of it at later time. Therefore the part of the trajectory after this step can be dropped.
This  is a Markov property analogous to the previous one, which lies behind the solution.
 Technically, the reduction of the number of particles continuing evolution can be performed by use of so called generalized Green functions introduced in
\cite{Brankov} and applied in \cite{PovPrS}, which in turn can be reduced to the determinantal process again. On the language of lattice paths this solution yeilds the probability of having a fixed bonds within the trajectories  selected particles.

\begin{figure}[tbp] \centering
\includegraphics[width=0.3\textwidth]{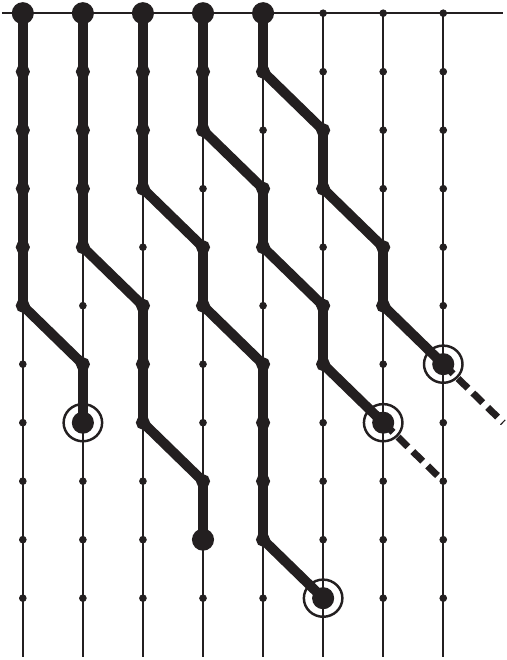}
\caption{ The extension considered in the present work. The time moments are not ordered, while the space positions are.
The steps of first two particles when they leave the zone
of possible interaction with the following particles are shown by dotted lines.}
\label{fig00}
\end{figure}

Our goal here is to unify all the previous achievements. Below we calculate the probabilities of  trajectories of selected particles to contain given points or bonds, as shown in in Fig.(\ref{fig00}). The range of point configurations we consider is wider then in the earlier solutions. Combination of two above Markov properties and use of the generalized Green function allow us to
remove time ordering completely. The tools we use, however, are applicable only when the spacial positions of the endpoints are strictly ordered in space. This is the only major constraint, which is nothing but the space-like condition described above \eqref{space-like ferrari coord}.

Though the ensemble of lattice paths gives a good pictorial representation of the problem, this language is not suitable for real calculations and presentation of the results, because the whole set of lattice paths is too big. To quantify the results we need a suitable probability space, where we could enumerate all our possible random outcomes. In the solutions mentioned above this was the set of particle coordinates $(\mathbb{Z})$, i.e. the lower horizontal line in  Fig.~\ref{fig0}(a), product of several such sets, i.e. subsequent horizontal lines in Fig.~\ref{fig0}(b), or the set of exit times enumerating the points at the vertical lines in Fig.~\ref{fig0}(b), respectively. Let us think about these lines as  the boundaries dividing the space-time plane into two parts. In all cases
the space-time trajectories of  particles go from one part to another right at the points we select. Therefore we can think of the probabilities under consideration as the probabilities for particle trajectories to go from the boundary at specified points. Known as exit probabilities
 such quantities are important in the extremal statistics \cite{KRB}. Exit probabilities is a convenient language to represent  most general correlation functions. To extend the range of space time configurations, we consider the boundaries of more general form: a broken line going from northeast to southwest by unit steps either vertical or horizontal, which divides the space-time plane into two parts. Consider now the space time trajectory of a single particle starting at the northwest part. Obviously, going from  the northwest  to  southeast, this trajectory will finally traverse the boundary. The question is, where will it happen?    We can enumerate the sites of the plane belonging to the boundary by a single generalized coordinate $\tau=t-x$, which runs over $\mathbb{Z}$. The value of $\tau$ corresponding to the site where the trajectory exits the boundary is a random variable, and its distribution $P(\tau<a)$ is the quantity of interest.  The probability distribution of particle coordinate at specified time moment and of the time the particle jumps from a specified site are particular cases of this general quantity.  Note that the exit occurs by two ways (down and down-right) from horizontal parts of the boundary and only  down-right  from vertical parts in the same way as above.

 The problem we address below is a direct generalization of one-particle picture described. We consider a collection of $m$ arbitrary boundaries, each with its own space-time coordinate  $\tau_i$ running in $\mathbb{Z}$, and enquire  about the joint distribution $P(\tau_1<a_1,\dots,\tau_m<a_m)$ of the coordinates of sites at which specified particles go from given boundaries, see Fig.\ref{fig000}. This construction allows one to remove any time ordering constraints and include into the scheme a possibility to consider both probability of particle being at a site and jumping from it. The geometric constraints on the boundaries from which the constraint on the accessible point configuration follow will be detailed in the next section.

\begin{figure}[tbp] \centering
\includegraphics[width=0.5\textwidth]{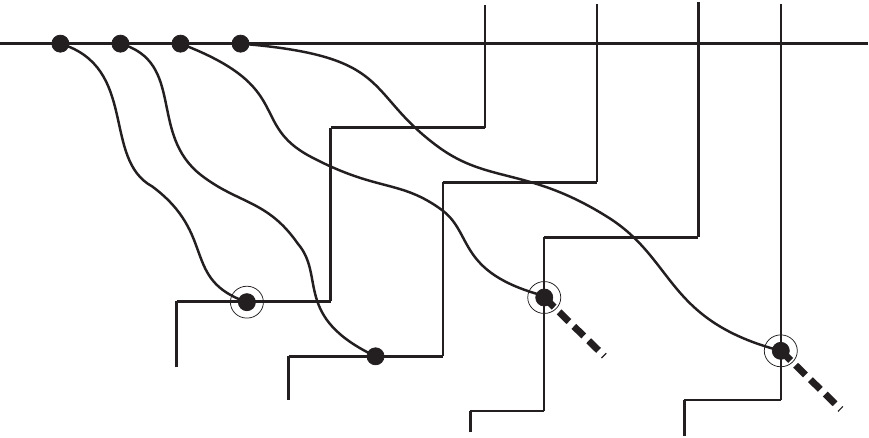}
\caption{Exit probabilities. The broken lines are the boundaries. An obligatory step forward must be added for exit to happen from vertical parts of the boundaries. The trajectory exiting from the horizontal part can continue in any of two ways. }
\label{fig000}
\end{figure}

After obtaining the results on exit probabilities we perform the scaling analysis of the formulas obtained. The lattice boundaries can be used to approximate smooth curves in the plane, and the selected points  are considered in the vicinity of a smooth path traversing these curves.  The main claim stemming from this analysis is that the large scale behaviour of the
 of exit probabilities is universal as far as the path under consideration do not violate the space-like constraint: the fluctuations of generalized exit coordinates of particles  starting from step initial conditions are described by Airy$_2$ ensemble in the same way as in purely spacial case.

The article is organized as follows. In the section \ref{Method and results} we give  definitions and formulate two main results of the paper:
exit probability distribution for trajectories of finite number of particles at the lattice   (Theorem \ref{distrib}) and its scaling limit (Theorem \ref{Airy_2}). In  section \ref{Determinantal point processes on the boundaries} we reformulate the TASEP in terms of signed determinantal process and prove theorem \ref{distrib} about exact form of the correlation function.  The section \ref{Asymptotic analysis of the correlation kernel} is devoted to asymptotic analysis of the results of previous sections, were we prove theorem \ref{Airy_2}.

\section{Method and results \label{Method and results}}

\subsection{Exit probabilities for particle trajectories on the space-time lattice}
To define exit probability for a single particle performing 1D asymmetric random walk,
consider a decomposition of the space-time 2D lattice into two complementary subsets
$\Omega\bigcup\bar{\Omega}=\mathbbm{Z}^2$. Given the random walk having started at  point $(x^0,t^0)\in \Omega$, the exit probability referring to  $\Omega$ is  a probability distribution of subsets of the boundary of $\Omega$
 from which the particle exits $\Omega$.  We will consider only sets having a property that once the particle has exited $\Omega$, it never returns there again. Then the  probability of exit from given point of the boundary does not depend on
 the global form of the boundary  of  $\Omega$.  Rather it is simply a product of  the probability for the particle trajectory to reach this point and the probability that the step from this points results an exit from $\Omega$. This is the case if  the boundary of $\Omega$ is defined in the following way.
\begin{definition}
\label{def1} The boundary $\mathcal{B}$ is an infinite countable subset of $%
\mathbbm{Z}^2$
\begin{equation}
\mathcal{B}=\{b(\tau)\in\mathbbm{Z}^2\}_{\tau \in \mathbbm{Z}},
\end{equation}
with the following staircase-like structure. Let $b(\tau)=(x,t)$. Then the next
point of the boundary will be either
\begin{equation}
b(\tau+1)=(x-1,t)  \label{boundarypoints}
\end{equation}
or
\begin{equation}
b(\tau+1)=(x,t+1),  \label{boundarypoints_v}
\end{equation}
for any $\tau\in \mathbbm{Z}$. A natural  integer variable $\tau$ increasing along the boundary from north-east to southwest  can be  chosen as $\tau=t-x, (x,t)\in \mathcal{B}$.
\end{definition}
Note that this construction ensures that the trajectory of a particle started in $\Omega$ eventually leaves $\Omega$ through the points of the boundary $\mathcal{B}$ with probability one and never returns there again. The probability distribution of the sets of these points is a
simplest example of the problem we address here. More generally one can consider a collection of embedded sets $\Omega_1\subset\Omega_2\subset \dots$, with boundaries $\mathcal{B}_1, \mathcal{B}_2, \dots$ and look for the joint distribution of successive exits  from these boundaries.

The idea of exit probabilities for  $N$ particles  undergoing the TASEP evolution on 1D lattice
generalizes the single-particle picture. Now we are interested
in how the trajectories of collection of interacting particles exit given sets. The quantity of interest is the joint distribution
of subsets of their boundaries at which exits occur.
Again, great simplification takes place i)~for such boundaries, that once the  trajectories  exited them they never return there again.
On the other hand we would like that for many particles ii) all possible configurations of exit points on the collection of boundaries
would be  assigned a probability measure in the same way as the points of the boundary in single-particle case.
The main tool which allows us to work with exit probabilities is the  Generalized Green Function (GGF).
Unlike purely spatial Green function used by other authors, the GGF allows us to work directly with space-time point configurations $(%
\bm{x,t})=((x_1,t_1),\dots,(x_N,t_N))$ belonging to the set of admissible
configurations defined by constraints
\begin{eqnarray}  \label{x order}
x_1>x_2>\dots>x_N \\
t_1\leq t_2 \leq \dots \leq t_N.  \label{t order}
\end{eqnarray}

For $N$ particles the concept of the boundary  can be generalized to $N$-boundary, which allows us meet (i) as well as  (ii).
\begin{definition}
\label{def2} Given boundary $\mathcal{B}$, the $N$-boundary $\bm{\mathcal{B}}%
_N \subset \{ 1, \dots, N \} \times \mathbbm{Z}$, is defined as a disjoint
union of $N$ copies of $\mathcal{B}$,
\begin{equation}  \label{N-boundary}
\bm{\mathcal{B}}_N=\bigsqcup_{k=1}^N \mathcal{B}_k,
\end{equation}
where the copy $\mathcal{B}_k=\{b_k(i)\}_{i\in \mathbbm{Z}^2}$ associated
with $k$-th particle is shifted by $(k-1)$ steps back with respect to the
first one in horizontal (spacial) direction of space-time plane,
\begin{equation}
b_k(i)=(x(i)-k+1,t(i)),
\end{equation}
$k=1,\dots,N$.
\end{definition}

The $N$-boundary is a generalization of the line with fixed time coordinate
and of the set of lines with fixed space coordinates, which where the
probability spaces used in \cite{BorodinFerrari,BFS} and in \cite{PovPrS}
respectively. Having started from an admissible point configuration, $N$
particle trajectories will reach given $N$-boundary after some evolution,
traverse it and go from some points of the $N$-boundary to continue the
evolution. Then, the non-crossing of the trajectories ensures that the
configuration of the departure points at the $N$-boundary is  admissible
as well.

To specify from which to which point sets the system can pass in course of
the TASEP evolution,  we also need a relation between subsets of $\{ 1,
\dots, N \} \times \mathbbm{Z}^2$.

\begin{definition}
Let $\Omega,\Omega^{\prime }\subset \{ 1, \dots, N \} \times \mathbbm{Z}^2$.
We say that relation
\begin{equation}
\Omega \prec \Omega^{\prime }
\end{equation}
holds, if for any $(x_k,t_k)\in \Omega$ and any $(x^{\prime }_k,t^{\prime
}_k)\in \Omega^{\prime }$
\begin{equation}
(x^{\prime }_k,t^{\prime }_k)\in \{(x,t):t\geq t_k\}\bigcup\{(x,t):x> x_k\}.
\end{equation}
\end{definition}

Note that the subindices denote the variable from the set $\{1,\dots,N\}$
and are associated with the number of a particle.

As it was explained in \cite{PovPrS}, a space-time trajectory of a particle
starting from a point preceding to a given boundary, eventually transverses
the boundary with probability one. The question we address is: What is the
probability for the trajectory to go from a given subset of the boundary?
More generally we address the same question to a  collection of
particles and a set of points at several boundaries.

\begin{figure}[tbp] \centering
\includegraphics[width=0.7\textwidth]{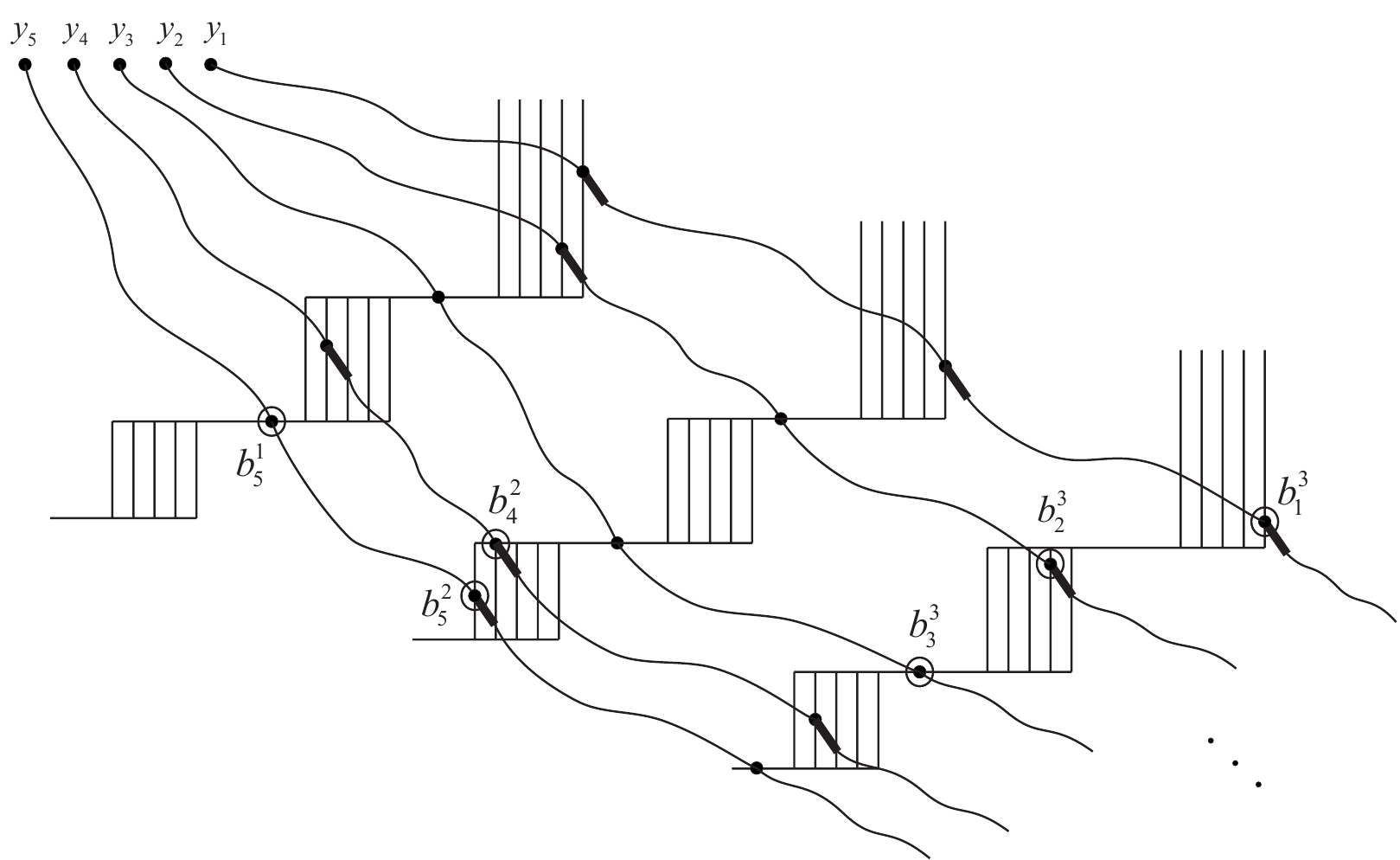}
\caption{Trajectories of five TASEP  particles traversing three $5$-boundaries. Black segments emphasize that  particles make a compulsory
step  forward at the sites belonging to vertical parts of boundaries from  which the exits occur. The exits included into correlation function with $N_1=N_2=5, N_3=4, N_4=3, N_5=2, N_6=1$ and $k_1=1,k_2=k_3=2,k_4=k_5=k_6=3$ are shown in circles.}
\label{fig1}
\end{figure}

To be specific, consider the TASEP evolution of $N$ particles governed by
the dynamical rules I-III. Let the initial configuration $\bm{x}^0$ be
defined by
\begin{equation}
x^0_i=-i+1,\,\,\, i=1,\dots,N.
\end{equation}
Let us fix a collection of $N$-boundaries, $\bm{\mathcal{B}}^1,\dots,%
\bm{\mathcal{B}}^m$, $m>0$, such that
\begin{equation}
\bm{x}^0\prec\bm{\mathcal{B}}^1\prec\dots\prec\bm{\mathcal{B}}^m.
\end{equation}
and fix the one-particle boundaries $\mathcal{B}_{N_1}^{k_1},\dots,\mathcal{B%
}_{N_l}^{k_l}$ within the $N$-boundaries. Here the upper indices $%
1=k_1\leq\dots\leq k_l= m$ refer to the number of $N$-boundary, the lower
indices, $N_l\leq\dots\leq N_1\equiv N $, to the particle number, and $l
\geq m$. We suggest that  at least one particle is fixed at each $N$-boundary,
i.e. either $k_{i+1}=k_i$ or $k_{i+1}=k_i+1$. We also require
that equality $N_i=N_{i+1}$ for some $i$ suggests that $k_{i+1}=k_i+1$, i.e.
two subsequent space-time points chosen for one particle should be put onto
subsequent $N$-boundaries, and no other particles with number less than $N_i$
can be fixed at the $N$-boundary $k_i$. Let space-time positions of points $%
b^k_n(i)$ within the corresponding boundary $\mathcal{B}_{n}^{k}$ be indexed
by index $i\in\mathbbm{Z}$ in the same way as in Defs. $(\ref{def1},\ref%
{def2})$. The quantity of interest is the joint probability distribution $P(i_1<a_1,\dots,i_l<a_l)$
 of the points $(b^{k_1}_{N_1}(i_1),%
\dots,b^{k_l}_{N_l}(i_l))$ from which the space-time trajectories of
particles $N_1,\dots,N_l$ make steps when leaving the boundaries $(\mathcal{B%
}^{k_1}_{N_1},\dots,\mathcal{B}^{k_l}_{N_l})$ respectively.

The first main result of the present paper can be stated as the following theorem.
\begin{theorem}
\label{distrib}  Under the above conditions the joint probability distribution of exit points is given by the Fredholm determinant
\begin{equation}
P(i_1<a_1,\dots,i_l<a_l)=\det(\mathbbm{1}-\eta_a K \eta_a)_{l^2(\{\mathcal{B}%
^{k_1}_{N_1},\dots,\mathcal{B}^{k_l}_{N_l}\})}
\end{equation}
with the kernel
\begin{eqnarray}
&&\!\!\!\!\!\!\!\! K(b_{N_i}^{k_i},b_{N_j}^{k_j})= \\
&=& \oint_{\Gamma_1}\frac{dv}{2\pi \mathrm{i}v }\oint_{\Gamma_{0,v}}\frac{dw%
}{2\pi \mathrm{i}w} \frac{\frac{(1-p(\frac{w-1}{w}))^{t_i}}{(1-p(\frac{v-1}{v%
}))^{t_j}} \frac{(w-1)^{N_i}}{(v-1)^{N_j}} \frac{w^{x_{N_i}}}{{v^{x_{N_j}}}}%
} {(w-v)(1/v+1/\pi_2-1)}  \notag \\
&-&\mathbbm{1}(N_2>N_1)\oint_{\Gamma_{0,1}}\frac{dw}{2\pi\mathrm{i} w^2}
\frac{(1-p(\frac{w-1}{w}))^{t_i-t_j} w^{x_{N_i}-x_{N_j}}} {%
(w-1)^{N_j-N_i}(1/v+1/\pi_2-1)}.  \notag
\end{eqnarray}
where $\eta_a=\mathbbm{1}(i_1\geq a_1)\times \cdots \times \mathbbm{1}(i_1\geq a_m) $,  $b_{N_i}^{k_i}=(x_{N_i},t_{N_i})\in\mathcal{B}^{k_i}_{N_i}$, $%
i,j=1,\dots,l$ and $\pi_2=1,p$ is the probability of step from the boundary
$\mathcal{B}_{N_j}^{k_j}$ at point $b_{N_j}^{k_j}$.
\end{theorem}

\subsection{Scaling limit of correlation functions \label{subsec: scaling limit}}
In the large scale the boundaries can be treated as approximations of continuous differentiable paths in the space-time plane.
Consider a scaling limit associated with sending to infinity a large parameter $L\to \infty$, as the time-space coordinates and particle numbers  measured at $L$-scale are fixed: $x/L$, $n/L$, $t/L=\mathrm{const.}$ Let us introduce  variable change $(x,t)\to(\chi,\theta)$:
\begin{eqnarray}
    \tau\equiv t-x&=&L \chi\label{tau}\\
    t+x&=&L\zeta(\chi, \theta)
\end{eqnarray}
As it was noted earlier the variable \eqref{tau} naturally enumerates points at the boundary. Correspondingly, the function $\zeta(\chi,\theta)$ defines a one-parameter family of curves spanning the whole space-time plane as $\theta$ varies in $\mathbbm{R}$. As the parameter $\chi$ runs in $\mathbbm{R}$, it defines a point at a particular curve corresponding to some fixed value of $\theta$.  The properties of $\zeta(\chi,\theta)$  follow from the properties of boundaries. Specifically, we suggest that
\begin{equation}\label{mod zeta}
    \left|\frac{\partial \zeta(\chi,\theta)}{\partial \chi}\right|\leq 1
\end{equation}
and
\begin{equation}\label{zeta-tau-theta}
  \left(\frac{\partial }{\partial \theta}-\frac{\partial }{\partial \chi}\right)\zeta(\chi,\theta) \geq 1.
\end{equation}
We now suppose that  for $k=1,\dots, m$  the boundaries $\mathcal{B}_1^k$ approximate the curves corresponding to fixed set $(\theta_1,\dots,\theta_m)$:
\begin{equation}\label{b_1^k}
    b_1^k([L\chi])=L\cdot\left( \frac{\zeta(\chi,\theta_k)-\chi}{2},\frac{\zeta(\chi,\theta_k)+\chi}{2}\right)+o(L^\sigma),
\end{equation}
where the notation $[\,\,]$ is for integer part of a real number and the correction term should not contribute on a characteristic fluctuation scale, i.e. $\sigma=1/3$. For technical purposes we will suggest that the correction term is uniform over the boundary. These boundaries correspond to the first particle. For general particle with  number   $n=[L\nu]$ we have to consider the boundary $\mathcal{B}^k_n$ shifting the spacial coordinate by $n-1$ steps backward:
\begin{equation}\label{b_n^k}
    b_n^k([L\chi])=L\cdot\left( \frac{\zeta(\chi,\theta_k)-\chi}{2}-\nu,\frac{\zeta(\chi,\theta_k)+\chi}{2}\right)+o(L^\sigma).
\end{equation}
Recall that on the large scale, $x\sim n\sim t\sim L\to \infty$, the trajectories of particles are deterministic, defined by the relation \eqref{det lim rel}. In terms of new variables the relation turns into
\begin{equation}\label{hydro zeta}
\sqrt{p(\zeta(\chi,\theta)+\chi)}-\sqrt{(\zeta(\chi,\theta)-\chi)}-\sqrt{2 q\nu}=0,
\end{equation}
which uniquely fixes  value of $\chi$ given those of $\theta$ and $\nu$, provided that the corresponding curve passes through the rarefaction fan defined by
\begin{equation}\label{rare-fan}
\chi\leq\zeta(\chi,\theta)\leq\frac{1+p}{1-p}\chi.
\end{equation}

Let us consider a path in $\theta\!-\!\nu$ plane:
\begin{equation}\label{theta-nu}
\begin{aligned}
\theta=\theta(r)\\
\nu=\nu(r)
\end{aligned}
, \qquad r\in\mathbbm{R},
\end{equation}
with differentiable functions $\theta(r)$ and $\nu(r)$, such that
\begin{equation}
    \frac{\partial \theta}{\partial r}\geq 0, \qquad \frac{\partial \nu}{\partial r}\leq 0 \label{theta nu constraint 1}
\end{equation}
and
\begin{equation}
    \frac{\partial \theta}{\partial r}-  \frac{\partial \nu}{\partial r}\geq 1.\label{theta nu constraint 2}
\end{equation}
We select $m$ points at the path, $r=r_1,\dots,r_m$,  so that the integers $N_1,\dots,N_m$ from Theorem~1.1 are given by
$N_i=[L\nu(r_i)]$, and $\theta_i=\theta(r_i)$. The inequalities \eqref{theta nu constraint 1} and \eqref{theta nu constraint 2} then guarantee that the constraints on   $k_1,\dots,k_m$ and $N_1,\dots,N_m$ from  Theorem~1.1 are satisfied and together with non-crossing of particle trajectories ensure that points of this path accessible for particle trajectories with nonzero probability form space-like configurations.

Substituting functions $\theta(r)$ and $\nu(r)$ into \eqref{hydro zeta} we obtain an equation, which, given $r$, can be  resolved with respect to $\chi$. For a  given path a unique solution exists for any $r$ within the range, in which the boundary corresponding to $\theta(r)$ passes trough the rarefaction fan \eqref{rare-fan}.  This solution is a monotonous function of $r$, which we  denote
 $\chi(r)$.  It defines the macroscopic deterministic location of the point, from where given particle exits given boundary, see Fig. \ref{Graphic2}.
 \begin{figure}[tbp]\centering
\includegraphics[width=0.7\textwidth]{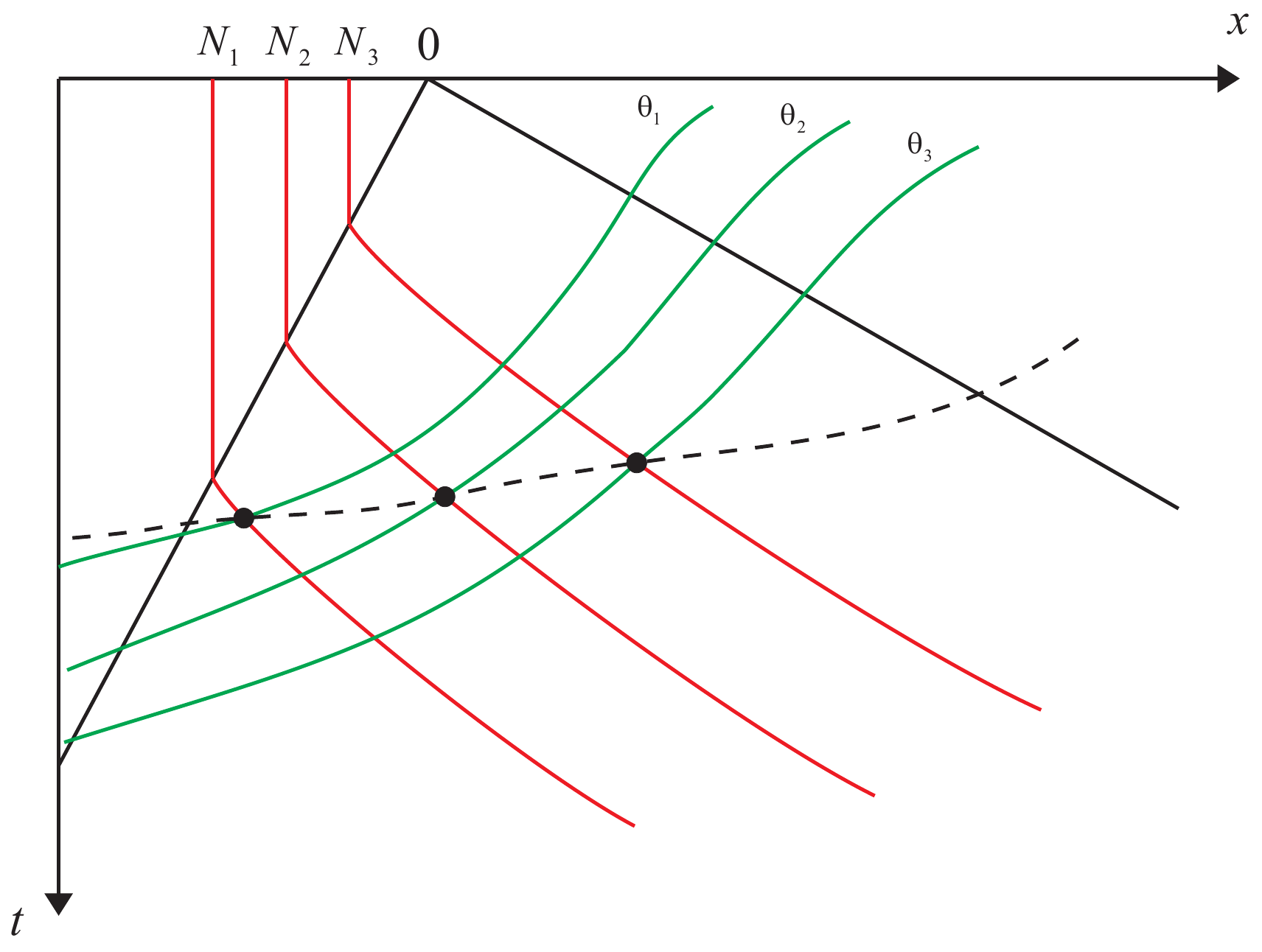}
\caption{Exit probabilities on a space-like path in $x$-$t$ plane. The wedge bounded by black straight lines is the rarefaction fan area.  The deterministic trajectories of particles with numbers $N_1,N_2,N_3$ are shown in red. The green lines are boundaries with coordinates $x=L((\zeta(\theta_i,\chi)-\chi)/2-\nu_i)$ and $t=L(\zeta(\theta_i,\chi)+\chi)/2$, where $i=1,2,3$, corresponding to three fixed values of $\theta$: $\theta_1<\theta_2<\theta_3$. Dashed line is the projection of the path $(\nu(r),\theta(r))$ to $x$-$t$ plane: $x=L((\zeta(r)-\chi(r))/2-\nu(r))$, $t=L(\zeta(r)+\chi(r))/2$. The black dots are the points where exits occur.}
\label{Graphic2}
\end{figure}
We are now turn to the fluctuations of these points referred to the boundaries and particle numbers separated by the distances of order of correlation length from each other. Suppose that
\begin{equation}
    r_i=r_0+u_iL^{-1/3}.
\end{equation}
The corresponding values of $\chi$ are given by their deterministic parts $\chi(r_i)$  plus a random variable of order of fluctuation
scale
\begin{equation}
   \chi_i=\chi(r_i)+\xi_L(u_i) L^{-2/3}.
\end{equation}
In what follows we show that the random variable $\xi_L(u)$ converges to the universal $\mathrm{Airy}_2$ process for a class boundaries, which can be approximated by \eqref{b_1^k}-\eqref{b_n^k}.
\begin{theorem}\label{Airy_2}
The following limit holds in a sense of finite-di\-men\-sion\-al distributions:
\begin{equation}
\lim_{L \to \infty} \xi_L=\kappa_f \mathcal{A}_2(\kappa_c u),
\end{equation}
where $\mathcal{A}_2$ is the $\mathrm{Airy}_2$ process characterized by multipoint distributions:
\begin{eqnarray}
\mathrm{Prob}(\mathcal{A}_2(u_1)<s_1,\dots,\mathcal{A}_2(u_m)<s_m)  \notag \\
=\det\left(\mathbbm{1}-\eta_{ s} K_{\mathrm{Airy_2}} \eta_{s})\right)_{L^2(\{n_1,\dots,n_m\}\times \mathbbm{R})}.
\end{eqnarray}
where in the r.h.s. we have the extended Airy kernel,
\begin{eqnarray}\label{airy kernel}
&K_{\mathrm{Airy}_2} &(\xi_1,\zeta_1; \xi_2,\zeta_2) \\
&&=\left\{
\begin{array}{ll}
\int_0^\infty d\lambda e^{\lambda(\xi_2-\xi_1)}\mathrm{Ai}(\lambda+\zeta_1)%
\mathrm{Ai} (\lambda+\zeta_2), & \xi_2\leq\xi_1 \\
-\int_{-\infty}^0 d\lambda e^{\lambda(\xi_2-\xi_1)}\mathrm{Ai}%
(\lambda+\zeta_1)\mathrm{Ai} (\lambda+\zeta_2), & \xi_2 > \xi_1%
\end{array}%
\right. ,  \notag
\end{eqnarray}
The model dependent constants  $\kappa_c$ and $\kappa_f$ defining the correlation and fluctuation scales respectively are given by
\begin{eqnarray} \label{kappa_h}
\kappa_c&=&\frac{p^{1/6}\left(\sqrt{\omega }-\sqrt{p \gamma } \right)^{-1/3} \left(\sqrt{p \omega } -\sqrt{\gamma }\right)^{-1/3} }{2 \gamma^{1/6} \omega^{1/6}    \left(\sqrt{p \gamma }(1+\zeta ^{(0,1)}(r_0)) +\sqrt{\omega }(1-\zeta ^{(0,1)}(r_0)) \right)} \\
&\times&\left[q \nu '\left(r_0\right) \left(\zeta \left(r_0\right)-\chi \left(r_0\right) \zeta ^{(0,1)}\left(r_0\right)\right) \left(\sqrt{p \omega } -\sqrt{\gamma }\right)^{-1}
\right.\notag \\
&&\hspace{3cm}-\left.\theta '\left(r_0\right)  \zeta ^{(1,0)}\left(r_0\right) \left(\sqrt{p \omega }-\sqrt{\gamma } \right)
\right]\notag \\ \notag \\
\label{kappa_t} \kappa_f&=&\frac{(\sqrt{p\omega}-\sqrt{\gamma})^{1/3}
\left(\sqrt{p \gamma }(1+\zeta ^{(0,1)}(r_0)) +\sqrt{\omega }(1-\zeta ^{(0,1)}(r_0)) \right)
}{2p^{1/6}\omega^{1/3}\gamma^{1/3}\left( \sqrt{\omega }-\sqrt{p \gamma } \right)^{2/3}},
\end{eqnarray}
where we denote $\zeta(r)\equiv \zeta(\theta(r),\chi(r))$,   $\zeta ^{(1,0)}(r_0)$ ($\zeta ^{(0,1)}(r_0)$) is the derivative of the function
$\zeta(\theta,\chi)$ with respect to the first (second) argument  at the point $(\theta(r_0),\chi(r_0))$ and parameters $\gamma$ and $\omega$ are those defined in \eqref{scal lim vars},   $\gamma=(\zeta(r_0)-\chi(r_0))/2$ and $\omega=(\zeta(r_0)+\chi(r_0))/2$.
\end{theorem}

The non-universal constants $\kappa_f$ and $\kappa_c$ are the most general ones for the TASEP with backward update. They depend not only
on the macroscopic space-time location defined by $\zeta(r_0)$ and $\chi(r_0)$, but also on the local slope and local density  of the boundaries at this point via the derivatives $\zeta ^{(0,1)}(r_0)$  and $\zeta ^{(1,0)}(r_0)$) respectively. Particular cases studied before can easily be restored from the expressions obtained. For example, for purely spacial boundary used for measuring particle coordinates at fixed time we can take $\zeta(\theta,\chi)=2t-\chi$, while the case of current correlation functions \cite{PovPrS} corresponds to $\zeta(\theta,\chi)=2x+\chi$.
For the space-like correlation functions of particle coordinates studied in \cite{BorodinFerrari,BFS} we take $\zeta(\theta,\chi)=2\theta-\chi$, and the tagged particle case \cite{ImamSasamoto} corresponds to $\nu'(r)=0$.
\section{Determinantal point processes on the boundaries \label{Determinantal point processes on the boundaries}}

\subsection{Single $N$-boundary.}

We first introduce the Generalized Green Function (GGF) using the
determinantal formula proposed in \cite{Brankov} and proved in \cite{PovPrS}, which generalizes
the formulae of simple Green function obtained in \cite{Sch1} for continuous time TASEP and generalized to the backward sequential update in
\cite{priezzhev_traject}

 Given two admissible
configurations
\begin{equation*}
\bm{b}^0\equiv (\bm{x}^0,\bm{t}^0)=((x_1^0,t_1^0),\dots,(x_N^0,t_N^0))
\end{equation*}
and
\begin{equation*}
\bm{b}\equiv (\bm{x},\bm{t})=((x_1,t_1),\dots,(x_N,t_N))
\end{equation*}
we define
\begin{equation}
G(\bm{b}|\bm{b}^0)=\det \left[F_{j-i}(b_i-b_j^0) \right]_{1\leq i,j\leq N}.
\label{GGFformula}
\end{equation}
where $(b_i-b_j^0)=(x_i-x_j^0,t_i-t_j^0)$ is componentwise extraction and
\begin{equation}  \label{FDisc}
F_n(x,t)=\left\{{\
\begin{array}{ll}
\frac{1}{2\pi \mathrm{i}}\oint_{\Gamma_0}\frac{dw}{w} \left(q+\frac{p}{w}%
\right)^t(1-w)^{-n}w^x, & t\geq0 \\
0, & t<0%
\end{array}%
} \right..
\end{equation}
For point $b(i)=(x,t)\in\mathcal{B}$ at the boundary, we introduce an exit
probability
\begin{eqnarray}
\pi^\mathcal{B}(b(i))= \left\{
\begin{array}{ll}
p, & \mathrm{if}\,\, b(i+1)=(x,t+1) \\
1, & \mathrm{if}\,\, b(i+1)=(x+1,t)%
\end{array}%
\right.,
\end{eqnarray}
and for $N$-point configuration $\bm{b}\in \bm{\mathcal{B}}$
\begin{equation}
\pi^{\bm{\mathcal{B}}}(\bm{b})=\prod_{k=1}^N \pi^{\mathcal{B}_k}(b_k),
\end{equation}
where the subscript $k$ specifies a boundary within the $N$-bounday, or the
associated particle. The function
\begin{equation}
\mathcal{G}(\bm{b}|\bm{b}^0)\equiv \pi^{\bm{\mathcal B}}(\bm{b}) G(\bm{b}|%
\bm{b}^0)
\end{equation}
gives probability for the space-time trajectories of particles to go away
from the boundary via the points of $\bm{b}$, given they started from $%
\bm{b}^0$.

We now show that this probability can be reinterpreted in terms of an
auxiliary signed determinantal point process on $\bm{\mathcal{B}}$. Consider
a signed measure on $\mathbb{Z}_{\geq \tau_0}\times\{1,\dots,N\}$,
\begin{eqnarray}  \label{det_measure}
\mathcal{M}\left(\mathcal{T}\right)=\frac{1}{Z_N} \prod_{n=0}^{N-1} \det[%
\phi_{n}(\tau^{n}_{i},\tau^{n+1}_{j})]_{i,j=1}^{ n+1} \det\left[%
\Psi_{N-i}^N(\tau_{j}^N)\right]_{i,j=1}^{N},
\end{eqnarray}
assigned to the sets of the form
\begin{equation}
\mathcal{T}=\bigsqcup_{1\leq n \leq N }
\{\tau^n_n,<\tau^n_{n-1},<\dots,<\tau^n_{1}\} \subset\mathbb{Z}_{\geq \tau_0}%
\times\{1,\dots,N\}.  \label{T}
\end{equation}
Here we define the function
\begin{equation}
\phi_n(z,y)=\left\{{\
\begin{array}{ll}
\pi^{\mathcal{B}_{n+1}}(b_{n+1}(y)), & y \geq z \\
0, & y < z%
\end{array}%
} \right.  \label{phi}
\end{equation}
for $x,y\in \mathbbm{Z}_{\geq \tau_0}$ and $b_{k}(y)\in \mathcal{B}_{k}$, $k=1,\dots,N$
and the function
\begin{equation}  \label{Psi_N}
\Psi^N_k(t)=(-1)^k \widetilde{F}_{-k}(b_{N}(t)-b^0_{N-k}),
\end{equation}
where
\begin{equation}  \label{tilde{F}_n}
\widetilde{F}_n(x,t)= \frac{1}{2\pi \mathrm{i}}\oint_{\Gamma_0}\frac{dw}{w}
\left(q+\frac{p}{w}\right)^t(1-w)^{-n}w^x.
\end{equation}
The integral representation holds for $t\in\mathbbm{Z}$. This is unlike $%
F_n(x,t)$ which coincides with $\widetilde{F}_n(x,t)$ when $t\geq 0$ and
vanishes at $t<0$, see (\ref{FDisc}). The numbers $\tau_i^j$ are integers bounded  by
 number $\tau_0$ from below. The number  $\tau_0$ is chosen so that $\Psi_k^N(\tau_0)=0$, which is always possible by construction of the boundaries.

We also introduce fictitious variables $\tau_{n}^{n-1}, 1\leq n\leq N$,
which are fixed to   $\tau_{n}^{n-1}=\tau_0$. Thus for any $\tau_i^j \in \mathbbm{Z}_{\geq \tau_0}$  we have
\begin{equation}
\phi_n(\tau_{n+1}^{n},\tau^{n+1}_j)\equiv \pi^{\mathcal{B}%
_{n+1}}(b_{n+1}(\tau^{n+1}_j))
\end{equation}
for $j=1,..,n+1$. The numbers $\tau_i^n$, $i=1,\dots,n$ are mapped to the
sites $b_n(\tau_i^n)$ on $\mathcal{B}_n$. Therefore the measure on the $N$%
-boundary $\bm{\mathcal{B}}$ is naturally defined as pushforward of $%
\mathcal{M}(\mathcal{T})$ under this mapping.

One can consider $b_n(\tau^n_j)$, $1\leq j\leq n\leq N$, as coordinates of
auxiliary fictitious particles indexed by $j$ leaving at the boundaries $%
\mathcal{B}_n$. These particles evolve as shown in Fig.\ref{fig3}.
 \begin{figure}[tbp]\centering
\includegraphics[width=0.7\textwidth]{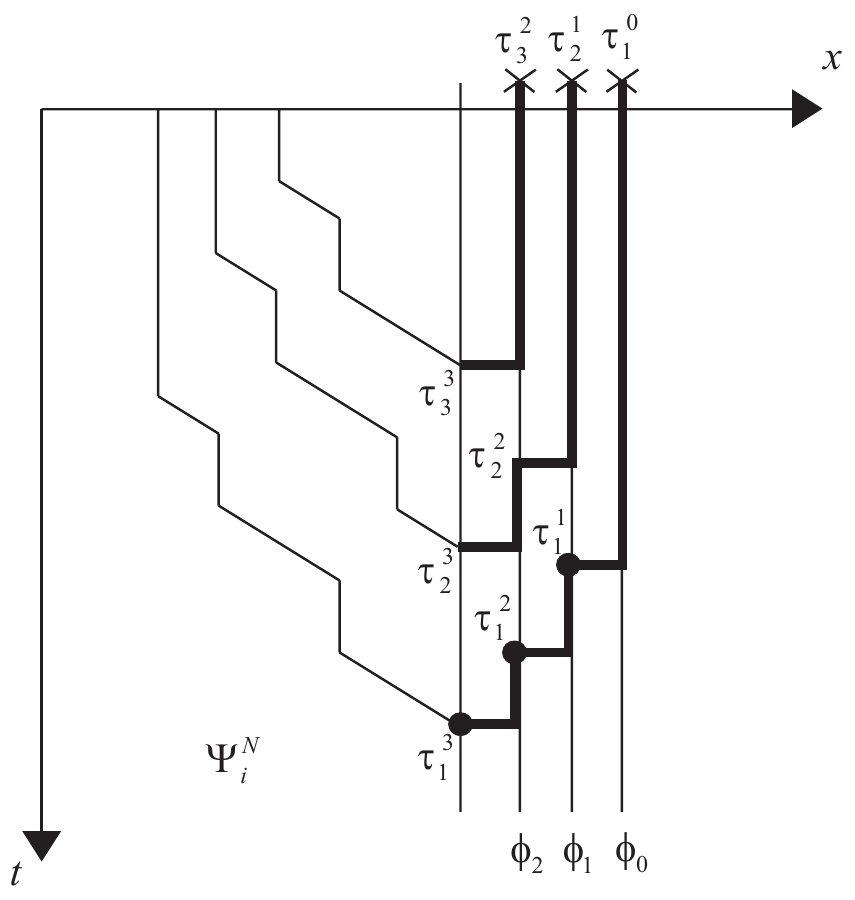}
\caption{Evolution of fictitious particles for $N=3$. The first stage from initial points to the boundary $\mathcal{B}_3$ is described by the functions $\Psi^N_i(\tau^N_k)$. The following stages from $\mathcal{B}_3$ to $\mathcal{B}_2$, from $\mathcal{B}_2$ to $\mathcal{B}_1$ and from $\mathcal{B}_1$ to $\mathcal{B}_0$ are encoded in $\phi_2,\phi_1$ and $\phi_0$ respectively. The coordinates of real TASEP particles are shown by black dots.}
\label{fig3}
\end{figure}
First, $N$
particles arrive from their initial state encoded in the functions
$\Psi^N_k(t)$ at the points of boundary $\mathcal{B}_N$ with numbers $%
\tau_1^N,\dots,\tau_N^N > \tau_0$. Then, they jump to the sites of the boundary $%
\mathcal{B}_{N-1}$ with the same numbers, and go up along the boundary $%
\mathcal{B}_{N-1}$ (from south-west to north-east, so that the number $\tau$
indexing position at the boundary decreases) any distance respecting mutual noncrossing
of particle trajectories. The weight of the jump between the boundaries,
outgoing from a site $b$, is $\pi^{\mathcal{B}_N}(b)$ and the weight of
going along the boundary is $1$ independently of the distance. The last ($N$-th)
particle  is forced to go to the resevoir $(\tau_N^{N-1}=\tau_0)$ and
disappear. The final positions of the other particles at the boundary $%
\mathcal{B}_{N-1}$ are denoted $\tau_1^{N-1},\dots,\tau_{N-1}^{N-1}$, from
which they jump to the boundary $\mathcal{B}_{N-2}$ with the weights $\pi^{%
\mathcal{B}_{N-1}}(\tau_1^{N-1}),\dots,\pi^{\mathcal{B}_{N-1}}(%
\tau_{N-1}^{N-1})$, e.t.c.. The process is repeated until the particle
number $1$ jumps from the point $b_1(\tau_1^1)$ of $\mathcal{B}_{1}$ and
disappears. This picture generalizes the auxiliary processes described for
the cases of constant time \cite{Sasamoto,BFPS,BorodinFerrari,BFS} and fixed
spacial coordinates \cite{PovPrS} to the case of general boundaries.

The fictitious particles are similar to vicious walkers (or free fermions), which
can be seen from the Karlin-McGregor-Lindstr\"om-Gessel-Viennot \cite{KM,L,GV}
determinantal form of the transition weights entering the product %
\eqref{det_measure} that ensure nonintersecting of their space-time
trajectories. The last determinant can be treated as integrated with given
initial distribution. Such a free fermionic structure allows calculation of
the correlation functions for fictitious particles, which turns out to be
determinantal. On the other hand, below we show that the joint distribution
of $N$ positions of the first fictitious particle obtained by integration of
the measure \eqref{det_measure} over the positions of the other particles
coincides with the Green function of TASEP. Thus the problem of correlations
in TASEP can be reduced to calculation of correlations between
noninteracting mutually avoiding fictitious particles.

To show that the GGF can be interpreted in terms of the measure $\mathcal{M}%
\left(\mathcal{T}\right)$ we prove the following proposition.

\begin{proposition}
\label{GGF-measure} Given $N$-boundary $\bm{\mathcal{B}}$, initial and final
configurations $\bm{b}^0 \prec \bm{\mathcal{B}}$ and $\bm{b} \subset %
\bm{\mathcal{B}}$  respectively, GGF $\mathcal{G}^{\bm{\mathcal{B}}}(\bm{b}|%
\bm{b}^0)$ associated with the boundary ${\bm{\mathcal{B}}}$ is a marginal
of the measure $\mathcal{M}$:
\begin{equation}  \label{marginal}
\mathcal{G}^{\bm{\mathcal{B}}}(\bm{b}|\bm{b}^0)=\mathcal{M}%
\left(\bigcup_{k=1}^N\{\tau_1^k=i_k\}\right),
\end{equation}
where $i_1,\dots,i_N$ determine the location of the points of $\bm{b}%
=(b_1(i_1),\dots,b_N(i_N))$ at corresponding boundaries within the $N$%
-boundary.
\end{proposition}

To prove this statement, one represents the GGF as a sum over the boundary
points in a way similar to that used for space variables in \cite%
{Sasamoto,BFPS,BorodinFerrari,BFS} and for time variables in \cite%
{Nagao,PovPrS}. The proof of the summation uses contiguous relations for the
values of the function  $\widetilde{F}_n(b(\tau))$ at adjacent points of the
boundary, which unify similar  relations for space and time variables.

\begin{lemma}
\label{cont rel}Let $b_k(\tau)$ be the point at the boundary $\mathcal{B}%
_k$ within the $N$-boundary $\bm{\mathcal{B}}$. Then the contiguous
relations hold for the function $\widetilde{F}_n(b_k(\tau))$
\begin{equation}  \label{cont}
\pi^\mathcal{B}(b_k(\tau))\widetilde{F}_n(b_k(\tau))=\widetilde{F}%
_{n+1}(b_{k-1}(\tau+1))-\widetilde{F}_{n+1}(b_{k-1}(\tau))
\end{equation}
\end{lemma}

\begin{proof}
The relation to be proved is in fact two contiguous relations for the
function $\tilde{F}_n(x,t)$,
\begin{eqnarray}
\tilde{F}_n(x,t) &=& \tilde{F}_{n+1}(x,t)-\tilde{ F}_{n+1}(x+1,t), \\
p \tilde{F}_n(x,t) &=& \tilde{F}_{n+1}(x+1,t+1)- \tilde{F}_{n+1}(x+1,t),
\end{eqnarray}
as one relation. The two latter relations follow from the integral
representation of the function $\tilde{F}_n(x,t)$.
\end{proof}

Then we have:
\begin{lemma}
\label{sasamoto sum} Given $N$-boundary $\bm{\mathcal{B}}$, initial and
final configurations $\bm{b}^0 \prec \bm{\mathcal{B}}$ and $\bm{b} \subset %
\bm{\mathcal{B}}$  respectively, the function $\mathcal{G}(\bm{b}|\bm{b}^0)$
can be represented as a sum:
\begin{eqnarray}
\mathcal{G}(\bm{b}|\bm{b}^0))&=&\sum_{ D}\prod_{1\leq i \leq n \leq N }\pi^{%
\mathcal{B}_n}\left( b_n(\tau^n_i)\right)  \notag \\
& \times& (-1)^{\frac{N(N-1)}{2}}\det [\widetilde{F}_{-N+1+i}(
b_N(\tau^N_{j+1})- b_{N-i}^0 )]_{i,j=0}^{N-1}  \label{2}
\end{eqnarray}
where $b_j(\tau_1^j)\equiv b_j, \;j=1,\dots,N,$ and the summation variables
take their values in the domain
\begin{equation}
D = \{\tau^j_i\in\mathbb{Z}_{\geq \tau_0},2\leq i\leq j \leq N | \tau^j_i\geq
\tau^{j-1}_i, \tau^j_i > \tau^{j+1}_{i+1}\}.  \label{3}
\end{equation}
\end{lemma}

\begin{proof}
Using the contiguous relation (\ref{cont}) the proof just follows the
similar proofs in \cite{Sasamoto,BFPS,BorodinFerrari,BFS,Nagao,PovPrS}. Note
that the lower summation bound $\tau_0$ is chosen such that the functions $\widetilde{F}$ under the determinant vanish at this point. Indeed this is true for  $\widetilde{F}_n(x,t)$ when $x>t$. By construction of boundaries it is always possible to find  suitable $\tau_0$ to ensure this inequality. To be specific we choose the maximal of these numbers.
\end{proof}

To complete the proof of proposition \eqref{GGF-measure} we need to show
that the summation over the domain $D$ can be replaced by the summation over
the sets of the form \eqref{T}.

\begin{lemma}\label{tilde_D}
The domain of summation in \eqref{2} can be replaced by
\begin{equation}  \label{tilde D}
\tilde{D}=\{\tau^j_i\in\mathbb{Z}_{\geq \tau_0},2\leq i\leq j \leq N |
\tau^j_i<\tau^{j}_{i-1}\}
\end{equation}
\end{lemma}

\begin{proof}
Apparently inequalities in \eqref{3} suggest those in \eqref{tilde D}. We
also need to show the converse: the measure \eqref{det_measure} is zero
everywhere in $\tilde{D}$ unless the inequalities from in \eqref{3} are
satisfied. The statement can be proved by reproducing the arguments from
\cite{PovPrS}.
\end{proof}

To find the correlation functions of the TASEP we first calculate the
correlation functions of the measure $\mathcal{M(T)}$. The functional form
of $\mathcal{M}(\mathcal{T})$ suggests that the correlation functions are
determinantal. Derivation of the correlation kernel was explained in
great detail in \cite{BFPS}. To proceed with the calculation, we introduce
convolution
\begin{equation}  \label{sec_phi}
\phi^{(n_1,n_2)}(x,y)=\left\{%
\begin{array}{ll}
(\phi_{n_1}*\phi_{n_1+1}*\dots*\phi_{n_2-1}) (x,y), & n_1<n_2 \\
0, & n_1\geq n_2%
\end{array}%
\right.,
\end{equation}
where $(a*b)(x,y)=\sum_{z\in\mathbf{Z}_{\geq \tau_0}}a(x,z)b(z,y)$, and
\begin{equation}  \label{Psi_n}
\Psi^n_{n-j}(\tau)=(\phi^{n,N}*\Psi^N_{N-j})(\tau).
\end{equation}

Note\, that\, in\, terms\, of\, the\, coordinates\, of\, fictitious\, particles\, function\, $\phi^{(n_1,n_2)}(x,y)$ is the transition weight between points at the boundaries $\mathcal{B}_{n_1}$ and $\mathcal{B}_{n_2}$. Hence, the points parameterized by the variables $x$ and $y$ in \eqref{sec_phi} live at $\mathcal{B}_{n_1}$ and $\mathcal{B}_{n_2}$, respectively, while the argument of $\Psi^n_{n-j}(\tau)$ in \eqref{Psi_n} lives on the boundary $\mathcal{B}_{n}$.\footnote{For single $N$-boundary this comment is not  essential  as  the points of different  boundaries within the same $N$-boundary have the same dependence on the index $\tau$ (specifically the value of $\pi(b(\tau)$)). Therefore, one could stick to,  for example, the boundary $\mathcal{B}_1$ shifting the spacial coordinate accordingly. Later, however, when we consider a sequence of $N$-boundaries, the information on what boundary the functions under consideration refer to becomes important.}

Consider functions
\begin{equation}  \label{space}
\{(\phi_0*\phi^{(1,n)})(\tau^0_1,\tau),\dots,(\phi_{n-2}*\phi^{n-1,n})(%
\tau^{n-2}_{n-1},\tau), \phi_{n-1}(\tau^{n-1}_n,\tau)\}.
\end{equation}
They are linearly independent and hence can serve as a basis of an $n$%
-dimensional linear space $V_n$. We construct another basis of $V_n$, $%
\{\Phi^n_j(\tau),j=0, \dots,n-1\}$, which is fixed by the orthogonality
relations
\begin{equation}  \label{orthogonality}
\sum_{\tau\in\mathbb{Z}}\Phi^n_i(\tau)\Psi^n_j(\tau)=\delta_{i,j}.
\end{equation}
Then, under the

\begin{description}
\item[Assumption (A)] : $\phi_n(\tau_{n+1}^n,\tau)=c_n \Phi_0^n(\tau)$ with
some $c_n\neq0$, $n=1,\dots,N$,
\end{description}

the kernel has the form
\begin{equation}  \label{Prop4.3.2}
K(n_1,\tau_1;n_2,\tau_2)=-\phi^{(n_1,n_2)}(\tau_1,\tau_2)+\sum_{k=1}^{n_2}
\Psi_{n_1-k}^{n_1}(\tau_1)\Phi_{n_2-k}^{n_2}(\tau_2).
\end{equation}

Applying repeatedly the convolution with $%
\phi_{N-1},\dots, \phi_j$  to  $\Psi^N_{N-j}(\tau)$  we obtain

\begin{lemma}
Given $N$-boundary $\bm{\mathcal{B}}$,  the functions $%
\Psi^n_j(\tau)$ have the following integral representation.
\begin{eqnarray}  \label{Psi_int}
\Psi^n_k(\tau)=\oint_{\Gamma_{0,1}}\frac{dw}{2\pi \mathrm{i} }\left(1-p\frac{%
w-1}{w}\right)^{t_n(\tau)-t^0_{n-k}} (w-1)^kw^{x_n(\tau)-x^0_{n-k}-1},
\end{eqnarray}
The contour
of integration $\Gamma_{0,1}$ encircles the poles $w=0,1$, leaving all the
other singularities outside.
\end{lemma}

To find basis $V_n$, we have to specify initial conditions. For usual step
initial conditions, the orthogonalization can easily be performed.

\begin{lemma}
Given step initial conditions, $b^0_k=(-k+1,0)$ for $k=1,\dots,N$, the
functions $\Psi^n_k(\tau)$ and $\Phi^n_j(\tau)$ satisfying (\ref%
{orthogonality}) are given by
\begin{eqnarray}  \label{Phi_int step}
\Psi^n_k(\tau)&=&\oint_{\Gamma_{0,1}}\frac{dw}{2\pi \mathrm{i}}\left(1-p%
\frac{w-1}{w}\right)^{t_n(\tau)} \left(\frac{w-1}{w}\right)^k w^{x_n(\tau)+n-2}, \\
\Phi^n_j(\tau)&=&\oint_{\Gamma_1}\frac{dv}{2\pi \mathrm{i}} \left(1-p\frac{%
v-1}{v}\right)^{-t_n(\tau)}\frac{ (v-1)^{-j-1}v^{j-x_n(\tau)-n} }{((1/\pi^{%
\mathcal{B}}(b(\tau))-1)v+1)},
\end{eqnarray}
where the contour of integration $\Gamma_{1}$ encircles the pole $v=1$
anticlockwise.
\end{lemma}

\begin{proof}
The function $\Psi^n_k(\tau)$ is obtained from \eqref{Psi_int} by an
explicit substitution of the step initial conditions. To prove the
orthogonality conditions (\ref{orthogonality}) one must evaluate the sum $%
\sum_{\tau \in \mathbbm{Z}}\Psi^n_k(\tau) \Phi^n_j(\tau)$. This is done by
an interchange of summation and integration. After successive summing the
geometric progressions for space-like and time-like parts of the boundary
and taking into account the pole at $v=w$, we obtain the desirable result.
To provide the convergence of the resulting sum we note that the choice of
contours ensures convergence of the sum for $\tau\to\infty$, while at the
lower limit the sum is truncated at $\tau=s$, so that $x_n(s)-t_n(s)-1=0$.
Obviously $\Psi^n_k(\tau)=0$ for $\tau < s$ because no poles remain inside
the integration contour for $k \geq 0$.
\end{proof}

Note that the form of $\Phi^n_j(\tau)$ depends on whether the site $b(\tau)$
belongs to time-like or space-like part of the boundary, which is reflected
in the term containing the exit probability in the denominator. Now we note
that the assumption \textbf{A} is fulfilled,
\begin{equation}
\Phi_0^n(\tau)=\pi^{\mathcal{B}_n}(b_n(\tau))=\phi_n(\tau_{n+1}^n,\tau),
\end{equation}
and we can write the kernel. The summation in (\ref{Prop4.3.2}) yields
\begin{eqnarray}  \label{sum phi psi}
&&\sum_{k=1}^{\infty}\Psi_{n_1-k}^{n_1}(\tau_1)\Phi_{n_2-k}^{n_2}(\tau_2) \\
&&= \oint_{\Gamma_1}\frac{dv}{2\pi \mathrm{i}v }\oint_{\Gamma_{0,v}}\frac{dw%
}{2\pi \mathrm{i}w} \frac{\frac{(1-p(\frac{w-1}{w}))^{t_{n_1}(\tau_1)}}{(1-p(\frac{%
v-1}{v}))^{t_{n_2}(\tau_2)}} \frac{(w-1)^{n_1}}{(v-1)^{n_2}} \frac{%
w^{x_{n_1}(\tau_1)}}{{v^{x_{n_2}(\tau_1)}}}} {(w-v)(1/v+1/\pi_2-1)} .  \notag
\end{eqnarray}
where $\pi^{\mathcal{B}}_2\equiv \pi^{\mathcal{B}}(b(\tau_2))$.

Observe that the function $\phi_n(x,y)$ can be written in the form
\begin{equation}  \label{phi_n_int}
\phi_n(\tau_1,\tau_2)=\oint_{\Gamma_{0,1}}\frac{dw}{2\pi\mathrm{i} w^2} \frac{%
(1-p(\frac{w-1}{w}))^{t_{n}(\tau_1)-t_{n+1}(\tau_2)} w^{x_{n}(\tau_1)-x_{n+1}(\tau_2)}} {%
(w-1)(1/v+1/\pi_2-1)}.
\end{equation}
After a few convolutions we have
\begin{eqnarray}  \label{phi^n_int}
\phi^{(n_1,n_2)}(\tau_1,\tau_2)&=& \mathbbm{1}(n_2>n_1) \\
&\times&\oint_{\Gamma_{0,1}}\frac{dw}{2\pi \mathrm{i} w^2} \frac{(1-p(\frac{%
w-1}{w}))^{t_{n_1}(\tau_1)-t_{n_2}(\tau_2)} w^{x_{n_1}(\tau_1)-x_{n_1}(\tau_2)}} {%
(w-1)^{n_2-n_1}(1/v+1/\pi_2-1)}.  \notag
\end{eqnarray}
Then we obtain:

\begin{proposition}
The correlation kernel of the measure $\mathcal{M}$, (\ref{det_measure}), is
\begin{eqnarray}  \label{cor_kernel}
&&\hspace{-0.5cm}K(n_1,\tau_1;n_2,\tau_2)=\\
&=& \oint_{\Gamma_1}\frac{dv}{2\pi \mathrm{i}v }\oint_{\Gamma_{0,v}}\frac{dw%
}{2\pi \mathrm{i}w} \frac{\frac{(1-p(\frac{w-1}{w}))^{t_{n_1}(\tau_1)}}{(1-p(\frac{%
v-1}{v}))^{t_{n_2}(\tau_2)}} \frac{(w-1)^{n_1}}{(v-1)^{n_2}} \frac{%
w^{x_{n_1}(\tau_1)}}{{v^{x_{n_2}(\tau_1)}}}} {(w-v)(1/v+1/\pi_2-1)}  \notag
\\
&-&\mathbbm{1}(n_2>n_1)\oint_{\Gamma_{0,1}}\frac{dw}{2\pi i w^2} \frac{(1-p(%
\frac{w-1}{w}))^{t_{n_1}(\tau_1)-t_{n_2}(\tau_2)} w^{x_{n_1}(\tau_1)-x_{n_1}(\tau_2)}} {%
(w-1)^{n_2-n_1}(1/w+1/\pi_2-1)},  \notag
\end{eqnarray}
where $\pi_2\equiv \pi^{\mathcal{B}}(b(\tau_2))$ and $x_n(\tau)=x(\tau)+n-1$.
\end{proposition}

Determinants of the above correlation kernel yield the correlation functions
of the measure $\mathcal{M}$, i.e. probabilities of point sets $\mathcal{T}$%
, (\ref{T}), having any given subsets. Then, using the inclusion-exclusion
principle, we can write down joint distribution
\begin{eqnarray}
\mathbf{P}=\mathcal{M}\left(\mathcal{T}\supset\{\tau_1^{n_1}\leq
a_1\}\bigcap \dots \bigcap \{\tau_1^{n_m}\leq a_m\}\right),
\end{eqnarray}
of sequences $\{\tau_1^{n_1},\dots,\tau_1^{n_k}\}$ for any fixed collection $%
1\leq n_1\leq\dots\leq n_m\leq N$, where $1\leq m\leq N$, in the form
\begin{eqnarray}  \label{fredholm one cascade}
\mathbf{P}=\det\left(\mathbbm{1}-\chi_a K
\chi_a)\right)_{l^2(\{n_1,\dots,n_m\}\times \mathbbm{Z})},
\end{eqnarray}
where Fredholm determinant is defined as a sum
\begin{equation}\label{fredholmsum}
\mathbf{P} =\sum_{n\geq 0}\frac{(-1)^n}{n!}\sum_{i_1,\dots,i_n=1}^m
\sum_{\tau_{1}>a_{i_1}}\dots\sum_{\tau_{n}>a_{i_n}}\det\{K(n_{i_k},%
\tau_k;n_{i_j}, \tau_j)\}_{k,j=1}^n
\end{equation}
and $\chi_a(n_i)(t)=\mathbbm{1}(t>a_i)$.  This distribution is the TASEP
correlation function of interest,
\begin{equation}
\mathbf{P} \equiv Prob\left(\{\tau_{n_1}\leq a_1\}\bigcap \{\tau_{n_2}\leq
a_2\}\bigcap \dots \bigcap \{\tau_{n_m}\leq a_m\}\right), \\
\end{equation}
and (\ref{fredholm one cascade}) is a particular case of the Theorem \ref%
{distrib} applied to the case of single $N$-boundary. Remarkably, the GGF
allowed us to treat very wide range of space-time point configurations ``in
one go'', in the same way as the fixed time and space cases were treated in
\cite{BFPS} and \cite{PovPrS}, respectively. Any admissible point
configuration can be processed in this way, when put to a suitable boundary.
The set of admissible configurations, however, does not exhaust all the
possibilities. It turns out that the time ordering constraint (\ref{t order}%
) can also be removed. To this end we apply a multicascade procedure,
similar to that used in \cite{BorodinFerrari}, to a sequence of $N$-boundaries.

\subsection{Multiple $N$-boundary case.}

Consider $m$ mutually distinct $N$-boundaries, $\bm{\mathcal{B}}%
^1\prec\dots\prec\bm{\mathcal{B}}^m$. In this section we derive the joint $l$%
-point probability distribution for positions $\mathrm{b}^{k_i}_{N_i}$ at which the trajectories of
particles $N_i$ depart from  boundaries $\mathcal{B}^{k_i}_{N_{i}}$, where $%
i=1,2,\dots,l$, and the indices $N_i$ and $k_i$ satisfy the assumptions of
 Theorem \ref{distrib}. We suggest that the space-time points within each
$N$-boundary are indexed independently by the indices $\tau_{i}(k)$, where
the subindex $i=1,\dots,N$ stands for the the number of the boundary $%
\mathcal{B}_i^k$ within the $N$-boundary $\bm{\mathcal{B}}^k$ and the
argument $k=1,\dots,m$ indexes $N$-boundaries. According to Defs. \ref{def1}%
, \ref{def2}, we first independently define an indexing order $%
b^{k}_{1}(\tau_{1}(k))$ for the first boundaries within each $N$-boundary
and then translate it to other $N-1$ boundaries by the corresponding left
shifts.

Given a fixed collection of integers $a_1,\dots,a_l$,  we are looking for
the joint probability
\begin{equation}
\bm{P}\equiv Prob\left( \bigcap_{i=1}^l \{\mathrm{b}_{N_i}^{k_i}={b}%
_{N_i}^{k_i}(\tau_{N_i}(k_i)) \}_{i=1}^l : \{\tau_{N_i}(k_i)\leq
a_i\}_{i=1}^l | \bm{b^0} \right)  \label{prob multi}
\end{equation}
for trajectories of particles $N_1,\dots,N_l$ to leave corresponding boundaries via points $%
\mathrm{b}_{N_i}^{k_i}$, located above (in terms of the corresponding
indices $\tau_i(k)$) the sites $b_{N_i}^{k_i}(a_i)$.

Similarly to the case of single $N$-boundary, our strategy is to represent
this probability distribution as a marginal of a signed determinantal
measure on a larger set. Suppose that  the set $(k_1, \dots,k_l)$ is of the
form
\begin{eqnarray}
k_{p_1}&=&\dots=k_{p_2-1}=1, \\
k_{p_2}&=&\dots=k_{p_3-1}=2,  \notag \\
&\vdots&  \notag \\
k_{p_m}&=&\dots=k_{l\quad\,\,}=m ,  \notag
\end{eqnarray}
which defines a collection of $m$ integers $1\equiv p_1\leq\dots\leq p_m\leq
l$. The quantity of interest can be given in terms of measure $%
\bm{\mathcal{M}}(\cdot)$ on point sets
\begin{equation}  \label{T bold}
\bm{\mathcal{T}}=\bigsqcup_{1\leq k \leq m }\mathcal{T}(k),
\end{equation}
where
\begin{equation}  \label{T(k)}
\mathcal{T}(k) = \bigsqcup_{N_{p_k}\leq n \leq N_{p_{k+1}}} \mathfrak{X}_n(k)
\end{equation}
and
\begin{equation}
\mathfrak{X}_n(k)=\{\tau^n_n(k),<\tau^n_{n-1}(k),<\dots,<\tau^n_{1}(k)\}
\subset \mathbbm{Z}_{\geq \tau_0(k)}.  \label{X_n(k)}
\end{equation}
 Then for $1\leq k\leq m$ the collections of integers $%
(\tau^n_n(k),\dots,\tau^n_{1}(k))$ define point configurations $%
(b_n^n(\tau^n_n(k)),\dots,b_1^n(\tau^n_{1}(k)))\subset \mathcal{B}_n^k$, which can be treated as coordinates of fictitious particles similarly to the single $N$-boundary case. The
pushforward of the measure $\bm{\mathcal{M}}$ under this mapping is a
measure on the collection of $N$-boundaries $\bm{\mathcal{B}}^1,\dots, \bm{\mathcal{B}}^m$. An explicit form of this
measure is as follows.
\begin{eqnarray}
\bm{\mathcal{M}}(\bm{\mathcal{T}})&=&\mathcal{N}^{-1}\det[%
\Psi^{N_1}_{N_1-l}(\tau^{N_1}_{k}(1))]_{1\leq k,l\leq N_1}  \label{mes1} \\
&\times&\prod\limits_{i=2}^m[\det[\mathcal{F}_{i,i-1}(\tau^{N_{p_i}}_l(i),%
\tau^{N_{p_i}}_k(i-1))]_{1\leq k,l\leq N_{p_i}}  \notag \\
&\times&\prod\limits_{n=N_{p_{i-1}}+1}^{N_{p_i}}\det[\phi_n(\tau^{n-1}_l(i),%
\tau^{n}_k(i))]_{1\leq k,l\leq n}],  \notag
\end{eqnarray}
where we define  functions
\begin{equation}
\begin{array}{l}
\Psi^{N_1}_{N_1-l}(\tau^{N_1}_{k}(1))= (-1)^{N_1-l}\widetilde{F}%
_{-N_1+l}(b_{N_1}^1(\tau^{N_1}_{k}(1))-b^0_{l}) \\
\mathcal{F}_{i,i-1}(\tau^{N_{p_i}}_l(i),\tau^{N_{p_i}}_k(i-1))= \pi^{%
\mathcal{B}_1^{i-1}}(b^{i-1}_1(\tau^{N_{p_i}}_{k}(i-1)))\\\hspace{4,3cm} \times \widetilde{F}%
_{0}(b^{i}_{1}(\tau^{N_{p_i}}_{l}(i))-b^{0,i}_{1}(\tau^{N_{p_i}}_k(i-1))),
\\
\phi_n(\tau^{n}_l(i),\tau^{n+1}_k(i))=\left\{{\
\begin{array}{c}
\pi^{\mathcal{B}_{n+1}^{i}}(b^i_{n+1}(\tau^{n+1}_k(i))),\,\,\,\,\,\,\,\,\,
\tau^{n+1}_k(i)\geq \tau^{n}_l(i), \\
0,\,\,\,\,\,\,\,\,\,\,\,\,\,\,\,\,\,\,\,\,\,\,\,\,\,\,\,\,\,\,\,\,\,\,\,\,\,%
\,\,\,\,\,\,\,\,\,\,\,\,\,\,\,\,\,\,\,\,\, \tau^{n+1}_k(i)<\tau^{n}_l(i),%
\end{array}%
} \right. \\
\phi_n(\tau^{n}_{n+1}(i),\tau^{n+1}_k(i))\equiv \pi^{\mathcal{B}%
_{n+1}^{i}}(b^i_{n+1}(\tau^{n+1}_k(i))),%
\end{array}
\label{not1}
\end{equation}
$\mathcal{N}$ is a normalization constant and for $1 \leq i \leq m$
\begin{equation}
b_{k}^{0,i}(\tau_k(i-1))=(x_k(i-1)+\frac{1-\pi^{\mathcal{B}%
_k^{i-1}}(b_k^{i-1})}{1-p}, t_k(i-1)+\frac{1-\pi^{\mathcal{B}%
_k^{i-1}}(b_k^{i-1})}{1-p}).  \label{b0}
\end{equation}
Lower cutoff $\tau_0(k)$ is separately chosen for every $N$-boundary $\bm{\mathcal{B}}_k$ in such a way, that any transitions to these points have zero measure. Specifically, as in the single $N$-boundary case considered in the previous subsection $\Psi^{N_1}_l(\tau)=0$  for any $\tau\leq \tau_0$. In addition $\mathcal{F}_{i,i-1}(\tau(i),\tau_0(i-1))=0$ for any $\tau(i)<\tau_0(i)$. Correspondingly,
the auxiliary variables $\tau_n^{n-1}(k)$ are fixed to $\tau_n^{n-1}(k)=\tau_0(k)$.

The relation between the correlation functions in TASEP and the measure $%
\bm{\mathcal{M}}$ is given by the following proposition.
\begin{proposition}
Consider the TASEP evolution starting with the initial conditions $\bm{b}%
^0=((x_1^0,t_1^0),\dots,(x_N^0,t_N^0))$, where $\bm{b}^0$ is admissible
configuration. Consider also $m$ mutually distinct $N$-boundaries, $%
\bm{b}^0  \prec \bm{\mathcal{B}}^1\prec\dots\prec%
\bm{\mathcal{B}}^m$. Let $(N_1, \dots,N_l)$ and $(k_1, \dots,k_l)$ be
collections of integers satisfying assumptions of the Theorem \ref{distrib}.
Then, the joint probability for space-time trajectories of  particles $%
N_1,\dots,N_l$ to go from $N$-boundaries  $\bm{\mathcal{B}}^{k_1},\dots,%
\bm{\mathcal{B}}^{k_l}$ via points $\mathrm{b}_{N_1}^{k_1},\dots, \mathrm{b}%
_{N_l}^{k_l}$, respectively, given the trajectories of all particles started
from the point configuration $\bm{b}^0$, is a marginal of the measure $%
\bm{\mathcal{M}}(\bm{\mathcal{T}})$ of the form
\begin{equation}
P\left(\bigcup_{i=1}^l \{\mathrm{b}_{N_i}^{k_i}=b_{N_i}^{k_i}(\tau_i(k))\} |%
\bm{b^0}\right) = \bm{\mathcal{M}}\left(\bm{\mathcal{T}}\supset
\bigcup_{i=1}^l\{\tau^{N_i}_1(k_i)=\tau_i(k),1\leq i\leq l\}\right).
\label{distrib_gen}
\end{equation}
\end{proposition}

\begin{proof}
We first note that instead of the $N$-boundaries $\bm{\mathcal{B}}^{1},\dots
,\bm{\mathcal{B}}^{m}$ we can consider auxiliary $N_{p_{i}}$-boundaries $%
\bm{\mathcal{B}}_{N_{p_{i}}}^{i}=(\mathcal{B}_{1}^{i},\dots ,\mathcal{B}%
_{N_{p_{i}}}^{i})$ for $i=1,\dots ,m$. This is possible because in TASEP the
trajectories of particles $p_{i}+1,\dots ,N$ do not influence the
trajectories $1,\dots ,p_{i}+1$, and no point of the former group is fixed after
(and on) $\bm{\mathcal{B}}^{i}$ within the correlation function %
\eqref{distrib_gen}. Given trajectories $1,\dots ,p_{i}$, the sum over all
realizations of the trajectories $p_{i}+1,\dots ,N$ amounts to one.
Therefore, after $\bm{\mathcal{B}}^{i}$ has been passed we can drop the
former evolution and consider only the latter. Thus, we first consider the transition
of $N_{1}$ particles from $\bm{b}^0$ to $\bm{\mathcal{B}}%
_{N_{p_{1}}}^{1}$, then the transition of $N_{p_{2}}$ particles from $%
\bm{\mathcal{B}}_{N_{p_{2}}}^{1}$ to $\bm{\mathcal{B}}_{N_{p_{2}}}^{2}$,
e.t.c. (see Fig \ref{fig1}). The probability of each transition is given by
corresponding $N_{p_{i}}$-particle Green function. To ensure the
admissibility of particle configurations within the Green function and keep
its probabilistic meaning we require that after each transition the
particles do leave the boundaries. This suggests that we insert a compulsory
step forward at the points belonging to vertical parts of the boundaries. To
this end, we supply each  step of this kind by the factor of $p$ and define the
starting points for every transition to be of the form (\ref{b0}). Finally,
the probability of interest, $P(\{b_{N_{i}}^{k_{i}}(\tau
_{N_{i}}(k_{i}))\}_{i=1}^{l}|\bm{b^0})$, is the following:
\begin{equation}
\begin{array}{l}
P\left( \bigcup_{i=1}^{l}\{\mathrm{b}_{N_{i}}^{k_{i}}=b_{N_{i}}^{k_{i}}(\tau
_{i}(k))\}|\bm{b^0}\right)  \\
=\sum\limits_{\Delta }\prod\limits_{j=1}^{m}G\left( \{b_{i}^{j}(\tau
_{i}(j))\}_{i=1}^{N_{p_{j}}}|\{b_{i}^{0,j}(\tau
_{i}(j-1))\}_{i=1}^{N_{p_{j}}}\right) \times \mathbf{\pi }^{\bm{\mathcal{B}}%
^{j}}(\mathbf{b}^{j})%
\end{array}
\label{JD1}
\end{equation}%
where $\{b_{i}^{0,1}(\tau _{i}(-1))\}_{i=1}^{N}\equiv \bm{b}^{0}$ and
the summation is over domain
\begin{eqnarray*}
\Delta =\{\tau _{j}(i)\in \mathbb{Z}_{\geq \tau_0(i)},\tau _{j}(i)> \tau _{j-1}(i),1\leq
j\leq N_{p_{i}},1\leq i\leq m\}\\ \setminus \{\tau _{N_{i}}(k_{i}),1\leq i\leq
l\}.
\end{eqnarray*}%
Using the determinantal formula of the GGF (\ref{GGFformula}), we have
\begin{equation}
\begin{array}{l}
P(\{b_{N_{i}}^{k_{i}}(\tau _{N_{i}}(k_{i}))\}_{i=1}^{l}|\bm{b^0})\\=\sum\limits_{\Delta }\prod\limits_{k=1}^{m}\mathbf{\pi }^{\bm{\mathcal{B}}%
^{k}}(\mathbf{b}^k)\det \left[ F_{i-j}\left(
b_{j}^{k}(\tau _{j}(k))-b_{i}^{0,k}(\tau _{i}(k))\right) \right] _{1\leq
i,j\leq N_{p_{k}}}. %
\end{array}
\label{JD2}
\end{equation}%
In what follows we are going to introduce auxiliary variables $\tau
_{i}^{j}(k)$ in the same way as we did for the case of single $N$-boundary,
with the only difference that there is a separate set for every $N$%
-boundary, indexed by an extra argument $k$. To proceed further we define
several domains of summation in these variables:
\begin{equation}
D_{i}=\{\tau _{k}^{j}(i)\in \mathbb{Z}_{\geq \tau_0(i)},1\leq k\leq j\leq N_{p_{i}}|\tau
_{k}^{j}(i)> \tau _{k+1}^{j}(i)\}
\label{Dom1}
\end{equation}%
\begin{equation}
\hat{D}_{i}=\{\tau _{k}^{j}(i)\in \mathbb{Z}_{\geq \tau_0(i)},1\leq k\leq j\leq
N_{p_{i+1}}-1|\tau
_{k}^{j}(i)> \tau _{k+1}^{j}(i)\}  \label{Dom2}
\end{equation}%
\begin{equation}
D_{i}^{\ast }=D_{i}\setminus \{\tau _{1}^{j}(i),N_{p_{i}}\leq j\leq
N_{p_{i+1}}\};\,\,\,\,\,\,\,\hat{D}_{i}^{\ast }=D_{i}^{\ast }\setminus \hat{D%
}_{i}  \label{Dom3}
\end{equation}%
\begin{equation}
D=\cup _{i=1}^{m}\hat{D}_{i}^{\ast }  \label{Dom4}
\end{equation}%
where we set $p_{m+1}\equiv l+1$ and $\tau _{1}^{k}(i)\equiv \tau _{k}(i)$
for $k=1,\dots ,N_{p_{i}}$ and $i=1,\dots ,m$.

Now we apply Lemmas \ref{sasamoto sum} and \ref{tilde_D} to each determinant under the product
in r.h.s. of (\ref{JD2}) to represent it as a sum over the auxiliary
variables:
\begin{eqnarray}
&\det&[F_{k-l}(b_{l}^i(\tau_l(i))-b_{k}^{0,i}(\tau_k(i-1)))]_{1 \leq k,l\leq
N_{p_i}}  \notag \\
&=&\sum\limits_{D_i^{*}}(-1)^{{\left(\!%
\begin{matrix}
N_{p_i} \\
2%
\end{matrix}%
\!\right)}} \det\left[\widetilde{F}_{-N_{p_i}+k}\left(b_{N_{p_i}}^i(%
\tau^{N_{p_i}}_{l}(i))-b_{k}^{0,i} (\tau^{k}_1(i-1))\right)\right]_{1 \leq
k,l\leq N_{p_i}}  \notag \\
&\times&\prod\limits_{n=0}^{N_{p_i}-1}\det\left[\phi_n(\tau^{n}_k(i),%
\tau_l^{n+1}(i))\right]_{1\leq k,l \leq {n+1}}.  \label{sasformula}
\end{eqnarray}
\begin{figure}[tbp]\centering
\includegraphics[width=0.7\textwidth]{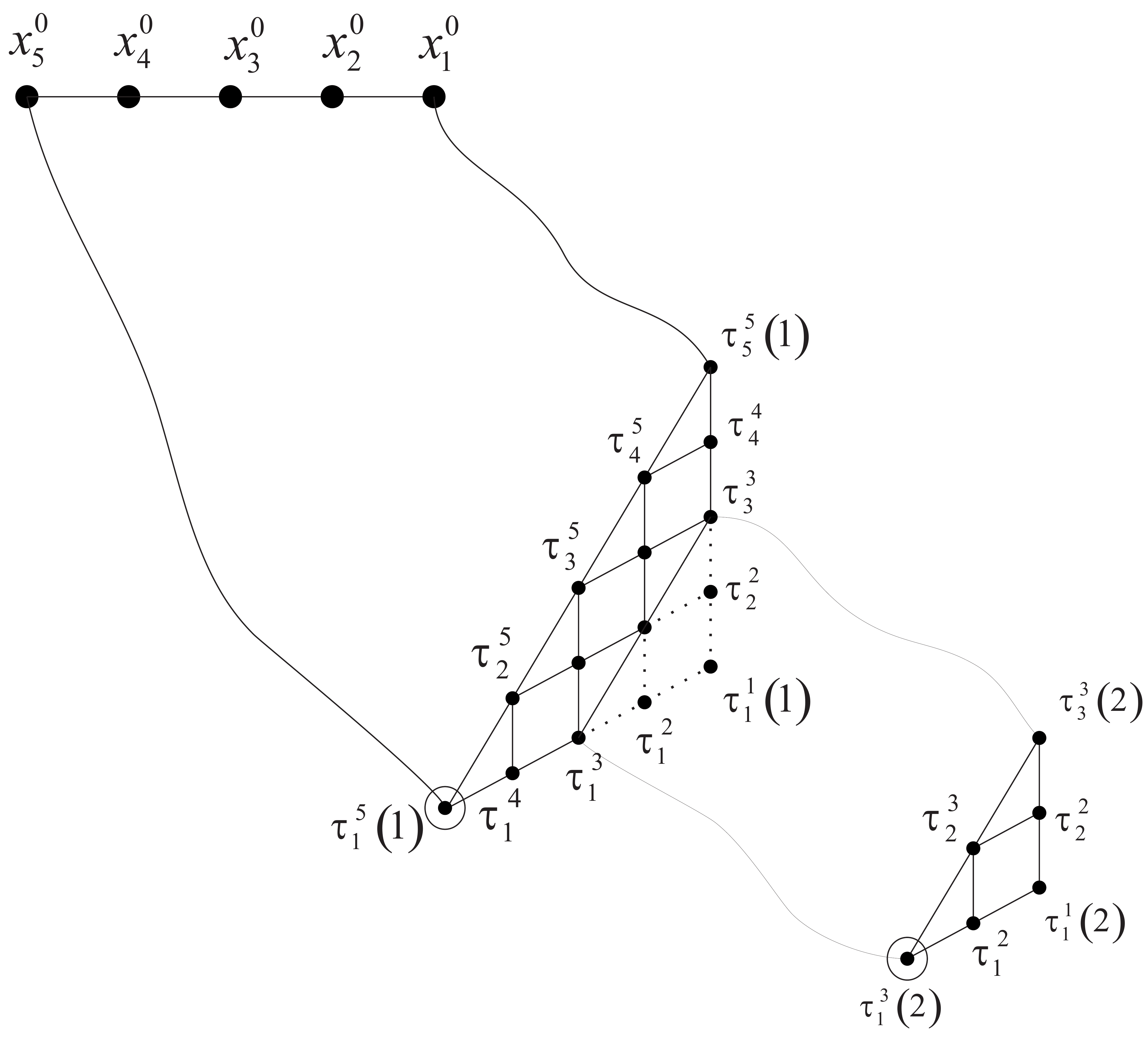}
\caption{Example of summation over the auxiliary variables $\tau^k_l(i)$ for two $N$-boundaries in the case $N_1=5
$ and $N_{p_2}=N_2=3$. The dotted lines show the variables, over  which the summation can be explicitly performed. The points shown by circles are fixed within the correlation function.
}
\label{fig2}
\end{figure}
The endpoints, $b_j^k(\tau)$, of part of the trajectories within a transition
between two $N_{p_i}$-boundaries are related to the starting points, $%
b_j^{0,k+1}(\tau)$, of the trajectories within the next transition by \eqref{b0}. The sums over
the range of these positions can be evaluated along with  a few sums in
auxiliary variables coupled to them (see Fig.\ref{fig2}):
\begin{equation}
\begin{array}{l}
\sum\limits_{\hat{D}_{i-1}}\det[\widetilde{F}_{-N_{p_i}+k}(b_{N_{p_i}}^i(%
\tau_{l}^{N_{p_i}}(i))-b_{k}^{0,i} (\tau^{k}_{1}(i-1)))]_{1 \leq k,l\leq
N_{p_i}} \\
\times\prod\limits_{n=0}^{N_{p_i}-1}\det[\phi_n(\tau^{n}_k(i-1),%
\tau_l^{n+1}(i-1))]_{1\leq k,l \leq n+1} \\
=(-1)^{{ \left(\!%
\begin{matrix}
N_{p_i} \\
2%
\end{matrix}%
\!\right)}}\det[\widetilde{F}_0(b_{1}^i(\tau^{N_{p_i}}_l(i))-b_{1}^{0,i}
(\tau^{N_{p_i}}_k(i-1)))]_{1 \leq k,l\leq N_{p_i}}.%
\end{array}
\label{JD4}
\end{equation}
The last identity can be proved by repeatedly applying formula
\begin{eqnarray}
&& \sum_{i=i_1}^{i_2}\pi^{\mathcal{B}}(b_k(i))\widetilde{F}_n(c_0-b_k(i))\\ \notag &&~~~~~~~=
\widetilde{F}_{n+1}(c_0-b_{k-1}(i_2+1))-\widetilde{F}%
_{n+1}(c_0-b_{k-1}(i_1)),
\end{eqnarray}
which is another form of Lemma \ref{cont rel}, where $c_0=(x_0,t_0)$ is a
pair of arbitrary constants.

The resulting expression for the joint distribution is
\begin{eqnarray}
&& \!\!\!\!\!\!\!\!\!\!\!\!\!\!\!\!\!\! P(\{b_{N_i}^{k_i}(\tau_{N_i}(k_i))\}_{i=1}^l |\bm{b^0}) \notag \\
& = &
const\sum\limits_{D}\det[\widetilde{F}_{-N_1+l}(b_{N_1}^1(%
\tau^{N_1}_{k}(1))-b^0_{l})_{1 \leq k,l \leq N_1}  \label{JD6} \\
& \times& \prod\limits_{i=2}^m\det[\widetilde{F}_{0}(b^{i}_{1}(%
\tau^{N_{p_i}}_{l}(i))-b^{0,i}_{1} (\tau^{N_{p_i}}_k(i-1)))]_{1 \leq k,l
\leq N_{p_i}}  \notag \\
& \times & \prod\limits_{n=N_{p_{i+1}}}^{N_{p_i}-1}\det[\phi_n(%
\tau^{n}_k(i),\tau_l^{n+1}(i))]_{1\leq k,l \leq n+1},  \notag
\end{eqnarray}
where we put $N_{p_{m+1}}\equiv 0$, $\hat{D}_m\equiv \emptyset$ and
therefore $\hat{D}^{*}_m\equiv D^{*}_m$.
\end{proof}


If we again appeal to the correspondence  with coordinates of  fictitious particles,
we see that the indices $\tau_n^k(i)$ define coordinates of particles at the boundary $\mathcal{B}_n^i$.
The functions $\widetilde{F}_0()$ under the product in \eqref{JD6} describe
the transitions between two subsequent $N$-boundaries, while  the functions $\phi(\tau^n,\tau^{n+1})$ are responsible for transitions between subsequent  $n$-th and $(n+1)$-th boundaries within the same $N$-boundary. Note that after we summed out part of  coordinates, some boundaries fell out of the consideration and only the  following remained: $$\mathcal{B}^{1}_{N_1},\dots, \mathcal{B}^{1}_{N_{p_2}},\mathcal{B}^{2}_{N_{p_2}},\dots,\mathcal{B}^{m-1}_{N_{p_{m}}},\mathcal{B}^{m}_{N_{p_m}},\dots,\mathcal{B}_{1}^{m}.$$
Therefore, it is  convenient to develop another enumeration, which counts only these boundaries.
As one can see, either upper index decreases or the lower one increases when going through the sequence.
Now we introduce new pair of indices, which distinguish  these two situation.
Each group within which the lower index does not change, such that for
some $i$ we have  $N_{p_{i}-1}>N_{p_i}=N_{p_i+1}=\dots=N_{p_i+c(n)-1}\equiv
n >N_{p_i+c(n)}$, is uniquely  characterized by number $n$, $1 \leq n \leq N$ and cardinality
$c(n)\in\{0,\dots,m+1\}$. This means that the particle number $n$
appears $c(n)$ times in the correlation function.
It is convenient to introduce a pair of indices $(n,a)$,
 where index $n$ is the number of particles arriving at given boundary and index $a$, $0\leq a \leq n-1$, labels the position of given  boundary within the group. Then,
instead of the notation $\tau_i^j(k) \in \mathcal{B}_j^k$ we use $\tau^{(n,a)}_i \in \mathcal{B}^{(n,a)}$, implying that
for each transition between two $N_p$-boundaries, in which the particle
number does not change, the second index $a$ increases by $1$, while in each
transition within single $N_p$-boundary, which effectively reduces the number of fictitious particles
 by one, index $n$ decreases by one. As a result, the r.h.s. of %
\eqref{mes1} can be rewritten in a more uniform way
\begin{eqnarray}
\mathrm{const}&\times&\det[\Psi^{N_1}_{N_1-l}(\tau^{(N_1,0)}_{k})]_{1\leq k,l\leq N_1}\notag
\\\label{mes2}
&\times&\prod\limits_{n=1}^{N_1}[\det[\phi_n(\tau^{(n-1,0)}_l,%
\tau^{(n,c(n))}_k)]_{1\leq k,l\leq n} \\
&\times&\prod\limits_{a=1}^{c(n)}\det[\mathcal{F}_{(n,a),(n,a-1)}(%
\tau^{(n,a)}_l,\tau^{(n,a-1)}_k)]_{1\leq k,l\leq n}.\notag
\end{eqnarray}

We are in position to apply Theorem 4.2 from \cite{BorodinFerrari}. It
states that the measure (\ref{mes2}) is determinantal and gives a recipe of
construction of the correlation kernel for given initial conditions.
Specifically, let us define function $\phi^{(n_1,a_1),(n_2,a_2)}$ of transition between the boundaries $\mathcal{B}^{(n_1,a_1)}$ and $\mathcal{B}^{(n_2,a_2)}$
\begin{eqnarray} \label{conv}
\phi^{(n_1,a_1),(n_2,a_2)}=\mathcal{F}_{(n_1,a_1),(n_1,0)}*\phi_{n_1+1} *%
\mathcal{F}_{(n_1+1,c(n_1+1)),(n_1+1,0)}\\ *\dots* \phi_{n_2}*\mathcal{F}%
_{(n_2,c(n_2)),(n_2,a_2)}, \notag
\end{eqnarray}
where we used a definition of convolution
\begin{equation}
(a*b)(x,y)=\sum_{z\in\mathbf{Z}_{\geq \tau_0^{(n,a)}}}a(x,z)b(z,y),
\end{equation}
 with the summation in $z$ performed over the points of the boundary $\mathcal{B}^{(n,a)}$, which is between the boundaries where the indices $x$ and $y$ live, and
\begin{equation}
\Psi^{(n,a)}_{n-l}=\phi^{((n,a),(N_1,0))}*\Psi^{(N_1,0)}_{N_1-l},
\label{bas}
\end{equation}
where $\Psi^{(N_1,0)}_{N_1-l}\equiv \Psi^{N_1}_{N_1-l}(\tau^{(N_1,0)})$. The argument $\tau^{(n,a)}$ of $\Psi^{(n,a)}_{n-l}$ lives on $\mathcal{B}^{(n,a)}$ due to the convolution with the function $\phi^{((n,a),(N_1,0))}$. For the cases when $c(n)=0$ we  formally define $\mathcal{F}_{(n,0),(n,0)}(x,y)=\delta_{x,y}$.

Consider matrix $M$ with matrix elements
\begin{equation}\label{M}
M_{k,l}=\left(\phi_{k}*\phi^{((k,c(k)),(N_1,0))}*\Psi^{(N_1,0)}_{N_1-l}\right)(\tau^{k-1}_{k}),
\end{equation}
where we can omit the dependence of $\tau^{k-1}_{k}(i)$  on the label $i$ of the  $N$-boundary $\bm{\mathcal{B}}_i$.
If the matrix M is invertible, the normalizing constant of the measure (\ref{mes2}) is equal
to $(\det M)^{-1}$. According to the Theorem 4.2 from \cite{BorodinFerrari}, the correlation kernel of (\ref{mes2}) is as follows
\begin{equation}
\begin{array}{l}
K(b^{(n_1,a_1)}(\tau_1),b^{(n_2,a_2)}(\tau_2))\\=
\sum\limits_{k=1}^{N_1}\sum\limits_{l=1}^{n_2}\Psi_{n_1-k}^{(n_1,a_1)}(\tau_1)
\left[M^{-1}\right]_{k,l}\left(\phi_{l}*\phi^{((l,c(l)),(n_2,a_2))}\right)(\tau^{l-1}_{l},\tau_2)\\
\hspace{6cm}
-\phi^{(n_1,a_1),(n_2,a_2)}(\tau_1,\tau_2).
\end{array}
\end{equation}
Furthermore, if the  matrix $M$ is upper triangular, the derivation of the kernel is significantly simplified.
In this case we construct the set of functions $\{\Phi_{k}^{(n,a)}\}$, which form a basis of the linear span of
the set
\begin{equation}
\{(\phi_1*\phi^{(1,c(1)),(n,a)})(\tau^0_1,\tau),\dots,
(\phi_{n}*\phi^{(n,c(n)),(n,a)})(\tau^{n-1}_{n},\tau)\},  \label{basispsi}
\end{equation}
fixed by orthogonality condition
\begin{equation}
\sum\limits_{\tau\in\mathbf{Z}}
\Psi_{i}^{(n,a)}(\tau)\Phi_{j}^{(n,a)}(\tau)=\delta_{i,j},\;\;\;\;\;%
\;i,j=0,...,n-1.
\end{equation}
Then the kernel takes the following form
\begin{equation}
K(b^{(n_1,a_1)}(\tau_1),b^{(n_2,a_2)}(\tau_2))=
\sum\limits_{k=1}^{n_2}\Psi_{n_1-k}^{(n_1,a_1)}(\tau_1)
\Phi_{n_2-k}^{(n_2,a_2)}(\tau_2)-\phi^{(n_1,a_1),(n_2,a_2)}.
\end{equation}
As a
result we have:

\begin{proposition}
Given densely packed initial conditions
\begin{equation}  \label{initial cond}
\bm{b}^0=(0,-1,\dots,-N+1)
\end{equation}
the correlation kernel of the determinantal measure (\ref{mes2}) has the
form
\begin{equation}
\begin{array}{l}
K(b^{(n_1,a_1)}_1;b^{(n_2,a_2)}_2)=\oint\limits_{\Gamma_1}\frac{dv}{2\pi iv}%
\oint\limits_{\Gamma_{0,v}}\frac{dw}{2\pi iw} \frac{\frac{(1-p\frac{w-1}{w}%
)^{t_1}(w-1)^{n_1}w^{x_1}}{(1-p\frac{v-1}{v})^{t_2}(v-1)^{n_2}v^{x_2}}} {%
(w-v)(1/v+1/\pi_2-1)} \\
\;\;\;\;\;\;\;\;\;\;\;\;\;\;\;\;\;\;\;\;\;\;\;\;\;\;\;\;\;
-1(n_2>n_1)\oint\limits_{\Gamma_{0}}\frac{dw}{2\pi iw^2} \frac{(1-p\frac{w-1%
}{w})^{t_1-t_2}w^{x_1-x_2}}{(w-1)^{n_2-n_1}(1/w+1/\pi_2-1)},%
\end{array}
\label{ker}
\end{equation}
where $b^{(n_i,a_i)}_i \equiv (x_i,t_i)\in \mathcal{B}^{(n_i,a_i)}$, $i=1,2$ and $\pi_2 \equiv\pi^{\bm{\mathcal{B}}%
^{k_2}}(b_{n_2}^{k_2})$.
\end{proposition}

\begin{proof}
We first introduce function $\hat{F}_n(b)$ defined by an integral representation, similar to the one of
$\widetilde{F}_n(b)$, with  different integration contour.
\begin{equation}  \label{hat{F}_n}
\hat{F}_n(x,t)= \frac{1}{2\pi \mathrm{i}}\oint_{\Gamma_{0,1}}\frac{dw}{w}
\left(q+\frac{p}{w}\right)^t(1-w)^{-n}w^x.
\end{equation}
One can check that this function has the following properties:
 \begin{eqnarray}
 \phi_n*\hat{F}_k&=&-\hat{F}_{k+1},\label{phi_n conv}\\
 \mathcal{F}_{(n,a),(n,a-1)}*\hat{F}_k&=&\hat{F}_k,\label{F_0 conv}
 \end{eqnarray}
and
\begin{equation} \label{F_k=0}
\hat{F}_k(x,t)=0, \quad \text{when}\quad k\leq 0 \quad \text{and} \quad -x>t.
\end{equation}
  Note that the choice of the contour $\Gamma_{0,1}$ ensures uniform convergence of convolution sums, which may extend to $x=-\infty$ and $t=\infty$. Therefore one can interchange summation and integration, from where the formulas  (\ref{phi_n conv},\ref{F_0 conv}) follow. The choice of the contour becomes relevant for $\hat{F}_k$ with positive $k$ as in this case there is a pole at $w=1$,  which must be placed inside the contour.
One also must keep in mind that the convolution with $\phi_n$  applied to the function
of a point at $\mathcal{B}^{(n,c(n))}$ results in a function of a point at $\mathcal{B}^{(n-1,0)}$, while the convolution with $\mathcal{F}_{(n,a),(n,a-1)}$ yields the transition from $\mathcal{B}^{(n,a-1)}$ to $\mathcal{B}^{(n,a)}$. Since $\hat{F}_k = \widetilde{F}_k$ for  $k<0$, $\Psi^{(N_1,0)}_{N_1-l}(\tau)=(-1)^{N_1-l}\hat{F}%
_{-N_1+l}(b^{(N_1,0)}(\tau)-b^0_{l})$, and hence, using (\ref{phi_n conv},\ref{F_0 conv}), we have
\begin{equation}
\Psi^{(n,a)}_{n-l}(\tau)=(-1)^{n-l}\hat{F}%
_{-n+l}(b^{(n,a)}(\tau)-b^0_{l}).
\end{equation}
Then, the elements of the matrix $M$ defined in \eqref{M} are
\begin{equation}
M_{k,l}=\phi_n* \Psi^{(k,c(k))}_{k-l}(\tau^{k-1}_k)=\hat{F}_{-k+l+1}(b^{(k-1,0)}(\tau_{k}^{k-1})-b^0_{l}).
\end{equation}
It follows from the definition of $\tau_k^{k-1}$ and formula \eqref{F_k=0} that $M_{k,l}=0$ when $k>l$ and
$M_{l,l}=1$. Therefore the matrix $M$ is invertible and upper triangular and we can
straightforwardly go to the orthogonalization procedure.

Substituting the initial conditions \eqref{initial cond} we obtain
\begin{equation}
\Psi_{k}^{(n,a)}(\tau)=\frac{1}{2\pi i}\oint\limits_{\Gamma_{0,1}}dw\left(1-p\frac{%
w-1}{w}\right)^{t^{(n,a)}}(w-1)^{k}w^{x^{(n,a)}+n-k-2},  \label{intpsi}
\end{equation}
where $(x^{(n,a)},t^{(n,a)})=b^{(n,a)}(\tau)$. It is not a surprise that this is the same function,
as the one obtained in the case of single $N$-boundary. Its argument lives on  single boundary $\mathcal{B}^{(n,a)}$, and the orthogonalization procedure referring to this boundary feels no difference with the previous subsection:
 \begin{equation}\label{intphi}
\Phi^{(n,a)}_j(\tau)=\frac{1}{2\pi i}\oint\limits_{\Gamma_1}dv\left(1-p\frac{v-1}{v}%
\right)^{-t^{(n,a)}} \frac{(v-1)^{-j-1}v^{j-n-x^{(n,a)}}}{(1/\pi^{\mathcal{B}%
}(b(\tau))-1)v+1}.
\end{equation}
Apparently, the double integral part of the kernel coincides with the one obtained in previous subsection as well.
We only need to derive  an explicit expression for $\phi^{(n_1,a_1),(n_2,a_2)}$. To this end we note that
we start the series of convolutions in \eqref{conv} with applying them  either to $\mathcal{F}_{(n_2,c(n_2)),(n_2,a_2)}$ or, if $c(n_2)=0$, to $\phi_{n_2}$. These functions can also be expressed in terms of $\hat{F}_k(x,t)$. Specifically, the expression for  $\phi_{n}$ obtained in the previous subsection is
\begin{equation}
\phi_{n}(\tau_1,\tau_2)=-\pi(b^{(n+1,c(n)),0}(\tau_2))\hat{F}_1(b^{(n,0)}(\tau_1)-b^{(n+1,c(n)),0}(\tau_2))
\end{equation}
and  from \eqref{not1}
\begin{equation}
\mathcal{F}_{(n_2,c(n_2)),(n_2,a_2)}(\tau_1,\tau_2)=\pi(b^{(n_2,a_2),0}) \hat{F}_0(b^{(n_2,a_2)}(\tau_1)-b^{(n_1,a_1),0}(\tau_2)).
\end{equation}
Therefore we can use formulas (\ref{phi_n conv},\ref{F_0 conv}) for convolutions,
which show that  the lower index of the function $F_k$ increases by one and the function itself picks up  a minus sign every time the number $n$ decreases by one. Finally we have
\begin{equation}
\phi^{(n_1,a_1),(n_2,a_2)}=(-1)^{n_1-n_2}\pi(b^{(n_2,a_2),0}_2)\hat{F}_{n_2-n_1}(b^{(n_2,c(n_2)}_2-b^{(n_2,a_2),0}_1),
\end{equation}
which again coincides with the expression obtained in single $N$-boundary case.
As a result we arrive at the kernel expression \eqref{ker}.
\end{proof}

Finally adopting the arguments from the end of the previous subsection for the collection of the boundaries $\mathcal{B}_{N_1}^{k_1},\dots,\mathcal{B}_{N_m}^{k_m}$ we arrive at the Fredholm determinant expression,
stated in the theorem \ref{distrib}. For the sake of mathematical rigor one would have to analyze the convergence of the series obtained (i.e. the properties of the operator $K$). Similar analysis however has been done in many papers and we address the reader to them \cite{Johansson,BFPS,BFP,BorodinFerrari}.

\section{Asymptotic analysis of the correlation kernel \label{Asymptotic analysis of the correlation kernel}}
Now we use the parametrization of the space-time plane discussed in subsection
\ref{subsec: scaling limit}. Below we evaluate the scaling limit of the correlation kernel,
suggesting that the arguments of the kernel are associated with a pair of boundaries and particle numbers
fixed by choosing two points at the path \eqref{theta-nu}-\eqref{theta nu constraint 2}  being at the distance of order of correlation length from each other.
\begin{lemma}
\label{scaling limit} Let us fix two points at the path  \eqref{theta-nu}-\eqref{theta nu constraint 2} in the $\theta\!-\!\nu$ plane
\begin{equation}
    r_i=r_0+u_i L^{-1/3},
\end{equation}
where $i=1,2$ and correspondingly set $n_i=[L\nu(r_i)]$ and  $\theta_i=\theta(r_i)$. Let us consider two boundaries $\mathcal{B}^{1}\equiv \mathcal{B}^{1}_1$ and $\mathcal{B}^{2}\equiv \mathcal{B}^{2}_1$ which approximate smooth curves according to
\eqref{b_1^k} with the parameters $\theta_1,\theta_2$ fixed above. Then $\eqref{b_n^k}$ define the curves approximated  by boundaries $\mathcal{B}^{1}_{n_1}$ and $\mathcal{B}^{2}_{n_2}$ corresponding to particles $n_1, n_2$, respectively.
For the coordinates $\tau_i$ of points on the boundary we also suggest the scaling
\begin{eqnarray}
\tau_i/L=\chi_i=\chi(r_i)+s_i L^{-2/3},\label{tau_i}
\end{eqnarray}%
with $u_{i},s_{i}$  fixed as $L\rightarrow \infty $ and the function $\chi(r)$ defined in the subsection \ref{subsec: scaling limit} as a deterministic part of the random variable $\chi$, obtained as a solution of the equation \eqref{hydro zeta} given $\theta(r)$ and $\nu(r)$.
Then
\begin{equation}
\lim_{L\rightarrow \infty }L^{1/3}K(b_{n_1}^{1}(\chi_1);b_{n_2}^{2}(\chi_2))\sim \kappa _{f} \Upsilon_{2} (\pi_2)
K_{\mathrm{Airy}_{2}}(\kappa_{c}u_{1},\kappa _{f}s_{1};\kappa
_{c}u_{2},\kappa _{f}s_{2}),
\end{equation}%
where in the r.h.s. we have the extended Airy kernel \eqref{airy kernel},
$\kappa_c$ and $\kappa_f$ are the  model-dependent constants (\ref{kappa_h},\ref{kappa_t})
and
\begin{equation}\label{upsilon}
\Upsilon_{2} (\pi_2)=\frac{2}{\left(\sqrt{p \gamma }(1+\zeta ^{(0,1)}(r_0)) +\sqrt{\omega }(1-\zeta ^{(0,1)}(r_0)) \right)} \!\times\! \left\{
\begin{matrix}
\sqrt{\omega } ,&\, \pi^{\mathcal{B}}(b(\tau_2))=1\\
 \sqrt{p \gamma  } ,&\,  \pi^{\mathcal{B}}(b(\tau_2))=p
\end{matrix}
\right.
\end{equation}
The sign $"\sim "$ means the equality up to the matrix conjugation, which does not affect matrix minors.
\end{lemma}
\begin{proof}
We  introduce the following functions
\begin{eqnarray}
f(w;\theta,\chi) &=&\frac{\zeta(\theta,\chi)+\chi}{2}\ln (q+p/w) \notag \\&+& \nu(r) \ln \left(\frac{w-1}{w}\right)   +\frac{ \zeta(\theta,\chi)-\chi }{%
2}\ln (w), \\
h(w) &=&\ln (1-1/w).
\end{eqnarray}%
To analyze the double integral part of the kernel $K_{0}$, we represent  it
as a sum
\begin{equation}\label{K_0}
  K_{0}(b_{n_1}^{1}(\chi_1);b_{n_1}^{2}(\chi_2))=\sum_{k=1}^{\infty }\Psi _{n_{1}-k}^{n_{1},\mathcal{B}^1}(\chi _{1})\Phi
_{n_{2}-k}^{n_{2},\mathcal{B}^2}(\chi _{2})
\end{equation}
 where
the functions $\Psi_{j}^{n,\mathcal{B}}$, $\Phi_{j}^{n,\mathcal{B}}$
are given in \eqref{intpsi},\eqref{intphi}. Note that,  instead of the index in the superscript characterizing the number of the boundary, we placed
the notation for the boundary explicitly, to reflect the dependence of the functions on the form of this boundary and not of the others (here $\mathcal{B}$  means the first particle boundary, while  the index $n$ shows that we have to shift it $n-1$ steps back in horizontal direction). In terms of above notations the integrals entering the summands become
\begin{eqnarray}
\Psi _{n_{1}-k}^{n_{1},\mathcal{B}^{1}}(\tau _{1}) &=&\oint_{\Gamma _{0,1}}\frac{dw}{%
2\pi \mathrm{i}w^{2}}e^{Lf(w;\,\theta_{1},\chi_1)+L^{1/3}zh(w) }, \\
\Phi _{n_{2}-k}^{n_{2},\mathcal{B}^{2}}(\tau _{2}) &=&\oint_{\Gamma _{1}}\frac{pdw}{%
2\pi \mathrm{i}}\frac{e^{-Lf(w;\,\theta_{2},\chi_2)-L^{1/3}zh(w))}}{%
(w-1)((1/\pi_2-1)v+1)},
\end{eqnarray}%
where $z=kL^{-1/3}$ and $\pi_2\equiv\pi^{\mathcal{B}^2}(b(\tau_2))$. To obtain the asymptotics of $K_{0}$, we first
evaluate the integrals for  $\Psi_{n_1-k}^{n_1,\mathcal{B}^1}(\tau_1)$ and $\Phi_{n_2-k}^{n_2,\mathcal{B}^2}(\tau_2)$ asymptotically as $L\rightarrow \infty $ and then
perform the summation.

Taking into account \eqref{tau_i}  one can approximate the function $f(w;\,\theta_{i},\chi_i)$ up to the terms of constant order by
\begin{equation}
    f(w;\,\theta_{i},\chi_i)=f_{r_i}(w)+L^{-2/3}s_i g(w),
\end{equation}
where we introduce the notations
\begin{equation}
    f_r(w)\equiv f(w;\,\theta(r),\chi(r))
\end{equation}
and
\begin{equation}
   g_r(w)=\frac{1}{2}\left[\left(\frac{\partial \zeta(r)}{\partial \chi}+1\right)\ln(q+p/w)+\left(\frac{\partial \zeta(r)}{\partial \chi}-1\right)\ln w\right],
\end{equation}
where $\zeta(r)\equiv \zeta(\theta(r),\chi(r))$.
The position of the double critical point of function $f_{r}(w)$, which
satisfies $f_{r }^{\prime }(w_{0})=f_{r }^{\prime \prime }(w_{0})=0
$ is
\begin{equation}
w_0(r)= 1+\sqrt{\frac{2\nu(r)}{q (\zeta(r) -\chi(r)  )}}.
\label{crit point}
\end{equation}%
Instead of the exponentiated functions we use their Taylor expansion at the
points $w_{i}\equiv w_{0}(r_{i})$, with $i=1,2$ for $\Psi
_{n_{1}-k}^{n_{1}}(\chi _{1})$ and $\Phi _{n_{2}-k}^{n_{2}}(\chi _{2})$
respectively.
\begin{eqnarray}
f_{r_{i}}(w) &\approx &f_{r _{i}}(w_{i})+\frac{1}{6}f_{r_0
}^{^{\prime \prime \prime }}(w_{0})(w-w_{i})^{3} \\
g_{r_i}(w) &\approx &g_{r_i}(w_{i})+g_{r_0}^{\prime }(w_{0})(w-w_{i}) \\
h(w) &\approx &h(w_{i})+h^{\prime }(w_{0})(w-w_{i})
\end{eqnarray}%
where in the coefficients of $w$-dependent terms we, without loss of accuracy, replace
$r_{i}$ and $w_{i}$ by $r_0 $ and $w_{0} \equiv w_{0}(r_0)$ respectively. We
substitute these expansion into the integrals, and choose steep descent contours such that
they approach the horizontal axis at the points $w_{1}$ and $w_{2}$ at the angles $\pm \pi
/3$ and $\pm 2\pi /3$ respectively. Changing the integration variables to $\xi
_{i}=(w-w_{i})L^{1/3}f^{\prime \prime \prime }(w_{0})/2$ we arrive at the integrals
defining the Airy functions:
\begin{equation}
\mathrm{Ai}(a)=\int_{\infty e^{-\mathrm{i}\pi /3}}^{\infty e^{\mathrm{i}\pi
/3}}\frac{dx}{2\pi \mathrm{i}}\exp \left( \frac{x^{3}}{3}-xa\right) .
\end{equation}%
As a result we have
\begin{eqnarray}
\Psi _{n_{1}-k}^{n_{1},\mathcal{B}^1}(\tau _{1}) &\approx &\frac{\exp\left({Lf_{r
_{1}}(w_{1})+L^{1/3}(s_{1}g_{r_1}(w_{1})+zh(w_{1}))}\right)}{w_{0}^{2}(Lf_{r_0
}^{^{\prime \prime \prime }}(w_{0})/2)^{1/3}}  \label{Psi_airy} \\
&\times &\mathrm{Ai}\left( \frac{zh^{\prime }(w_{0})-s_{1}g^{\prime }_{r_0}(w_{0})%
}{(f_{r_0}^{^{\prime \prime \prime }}(w_{0})/2)^{1/3}}\right)   \notag \\
\Phi _{n_{2}-k}^{n_{2},\mathcal{B}^2}(\tau _{2}) &\approx &\frac{ \exp\left({-Lf_{r
_{2}}(w_{2})-L^{1/3}(s_{2}g_{r_2}(w_{2})+zh(w_{2}))}\right)}{(w_{0}-1)({(1/\pi^{\mathcal{B}%
}(b(\chi_2))-1)w_0+1})(Lf_{r_0
}^{^{\prime \prime \prime }}(w_{0})/2)^{1/3}} \\
&\times &\mathrm{Ai}\left( \frac{zh^{\prime }(w_{0})-s_{2}g^{\prime }(w_{0})%
}{(f_{r_0}^{^{\prime \prime \prime }}(w_{0})/2)^{1/3}}\right)   \notag
\end{eqnarray}%
The summation over $k$ can be replaced by an integration over $z$. To
perform the summations we use one more expansion:
\begin{equation}
h(w_{i})=h(w_{0})-h^{\prime }(w_{0})w_{0}^{\prime }(r_0 )u_{i}L^{-1/3}+O(L^{-2/3}).
\end{equation}%
Finally, taking into account that $h'(w_0)=1/(w_0(w_0-1))$, we obtain
\begin{eqnarray}
\sum_{k=1}^{\infty }&&\!\!\!\!\Psi _{n_{1}-k}^{n_{1},\mathcal{B}^1}(\tau _{1})\Phi
_{n_{2}-k}^{n_{2},\mathcal{B}^2}(\tau _{2})   \label{Airy sum} \\
&\approx& \Upsilon_{2} (\pi_2) L^{-1/3}  \kappa_{f}e^{({L(f_{\nu _{1}}(w_{1})-f_{\nu
_{2}}(w_{2}))+L^{1/3}(s_{1}g(w_{1})-s_{2}g(w_{2}))})}  \notag \\
&\times& \int_{0}^{\infty }d\lambda e^{\lambda \kappa%
_{c}(u_{2}-u_{1})}\mathrm{Ai}\left( \lambda +\kappa%
_{f}s_{1}\right) \mathrm{Ai}\left( \lambda +\kappa%
_{f}s_{2}\right) ,  \notag
\end{eqnarray}%
where
\begin{eqnarray}
\kappa_{c} &=&\frac{w_{0}^{\prime }(\theta )f_{r_0
}^{^{\prime \prime \prime }}(w_{0})^{1/3}}{2^{1/3}} \\
\kappa_{f} &=&-\frac{2^{1/3}g^{\prime }(w_{0})}{f_{r_0}^{^{\prime \prime
\prime }}(w_{0})^{1/3}}
\end{eqnarray}
and
\begin{equation}
    \Upsilon_{2} (\pi_2)= - \left[g'(w_0)w_0 ({(1/\pi^{\mathcal{B}%
}(b(\tau_2))-1)w_0+1})\right]^{-1}
\end{equation}

Substituting
\begin{eqnarray}
w'_0(r_0)&=&\frac{\sqrt{p} }{2 \gamma q   \left(\sqrt{p \gamma }(1+\zeta ^{(0,1)}(r_0)) +\sqrt{\omega }(1-\zeta ^{(0,1)}(r_0)) \right)} \\
&\times&\left[q \nu '\left(r_0\right) \left(\zeta \left(r_0\right)-\chi \left(r_0\right) \zeta ^{(0,1)}\left(r_0\right)\right) \left(\sqrt{p \omega } -\sqrt{\gamma }\right)^{-1}
\right.\notag \\
&&\hspace{3cm}-\left.\theta '\left(r_0\right)  \zeta ^{(1,0)}\left(r_0\right) \left(\sqrt{p \omega }-\sqrt{\gamma } \right)
\right]\notag \\
f'''_{r_0}(w_0)&=&\frac{2 q^3\gamma ^{5/2}}{p \sqrt{\omega }\left(\sqrt{\omega }-\sqrt{p \gamma } \right) \left(\sqrt{p \omega } -\sqrt{\gamma }\right) }\\
g'(w_0)&=&\frac{q\sqrt{\gamma }
\left(\sqrt{p \gamma }(1+\zeta ^{(0,1)}(r_0)) +\sqrt{\omega }(1-\zeta ^{(0,1)}(r_0)) \right)
}{2\sqrt{p \omega }\left(\sqrt{p \gamma } - \sqrt{\omega }\right)}
\end{eqnarray}
we have (\ref{kappa_h},\ref{kappa_t},\ref{upsilon}).

Let us now evaluate the second part of the kernel given by the single
integral, which can be written as
\begin{equation*}
I=\oint \frac{dz}{2\pi \mathrm{i}}\frac{\exp\left[L(f_{r _{1}}(z)-f_{r
_{2}}(z))+L^{1/3}(s_{1}g_{r_1}(z)-s_{2}g_{r_2})\right]}{z({(1/\pi^{\mathcal{B}%
}(b(\tau))-1)z+1})}
\end{equation*}%
The critical point of the exponentiated function is found to be
\begin{equation}
z_{c}=w_{0}\equiv w(r_0 )
\end{equation}%
Then using the Taylor expansions we show that
\begin{eqnarray}
f_{r_{1}}(w)-f_{r_{2}}(w) &\approx &f_{r_{1}}(w_{1})-f_{r_{2}}(w_{2}) \\
&+&\frac{f_{r_0}^{\prime \prime \prime }(w_{0})(u_{2}^{3}-u_{1}^{3})w_{0}^{\prime
}(r_0)^{3}}{6L}  \notag \\
&+&\frac{f_{r_0}^{\prime \prime \prime }(w_{0})(u_{1}^{2}-u_{2}^{2})w_{0}^{\prime
}(r_0)^{2}}{2L^{2/3}}(z-w_{0})  \notag \\
&+&\frac{f_{r_0}^{\prime \prime \prime }(w_{0})(u_{2}-u_{1})w_{0}^{\prime }(r_0 )}{%
2L^{1/3}}(z-w_{0})^{2}  \notag
\end{eqnarray}%
and
\begin{eqnarray}
(s_{1}g_{r_1}(z)-s_{2}g_{r_2}(z)) &\approx &s_{1}g_{r_1}(w_{1})-s_{2}g_{r_2}(w_{2}) \\
&+&(u_{2}s_{2}-u_{1}s_{1})L^{-1/3}g^{\prime }_{r_0}(w_{0})w^{\prime }(\nu )  \notag
\\
&+&(s_{1}-s_{2})g^{\prime }_{r_0}(w_{0})(z-w_{0})  \notag
\end{eqnarray}%
Substituting these expansions into the integral and integrating along the
vertical line crossing the horizontal axis at $w_{0}$ we obtain:
\begin{eqnarray}
I &=&L^{-1/3} \Upsilon_{2} (\pi_2) \kappa_{f}e^{L(f_{c _{1}}(w_{1})-f_{c
_{2}}(w_{2}))+L^{1/3}(s_{1}g(w_{1})-s_{2}g(w_{2}))} \\
&&\frac{e^{\frac{\kappa_{c}^{3}(u_{2}^{3}-u_{1}^{3})}{3}-\frac{(%
\kappa_{c}^{2}(u_{1}^{2}-u_{2}^{2})-\kappa%
_{f}(s_{1}-s_{2}))^{2}}{4\kappa_{c}(u_{2}-u_{1})}-
\kappa_{c}\kappa_{f}(s_{2}u_{2}-s_{1}u_{1})}}{\sqrt{4\pi
\kappa_{c}(u_{2}-u_{1})}}.
\end{eqnarray}%
One can see that the first line of this expression exactly coincides with
the factor before the integral in (\ref{Airy sum}). Furthermore, its
exponential part does not change the value of the determinants, so that it
can be omitted. The second part can be rewritten using the formula from \cite%
{johansson2}
\begin{eqnarray}
&&\frac{1}{\sqrt{4\pi (\tau ^{\prime }-\tau )}}e^{-(\xi -\xi ^{\prime
})^{2}/4(\tau ^{\prime }-\tau )-(\tau ^{\prime }-\tau )(\xi +\xi ^{\prime
})/2+(\tau ^{\prime }-\tau )^{3}/12} \\
&=&\int_{-\infty }^{\infty }e^{-\lambda (\tau -\tau ^{\prime })}\mathrm{Ai}%
(\xi +\lambda )\mathrm{Ai}(\xi ^{\prime }+\lambda )d\lambda ,
\end{eqnarray}%
where we should set $\tau =\kappa_{c}u_{1},\tau ^{\prime }=%
\kappa_{c}u_{2},\xi =\kappa_{f}s_{1},\xi ^{\prime }=\kappa_{f}s_{2}$. As a result we
obtain the Airy extended kernel
\begin{equation*}
L^{-1/3}\Upsilon_{2} (\pi_2)\kappa_{f}K_{\mathrm{Airy}_{2}}(\kappa%
_{c}u_{1},\kappa_{f}s_{1};\kappa_{c}u_{2},%
\kappa_{f}s_{2})
\end{equation*}
\end{proof}

To finish  the proof of the theorem \ref{Airy_2} one has to prove the uniform convergence of the kernel  in bounded sets  and that the part of the sum \eqref{fredholmsum} coming from the complement to these sets is negligible while the bound is uniform in $L$. For similar proofs we address the reader to \cite{Johansson,BFPS,BFP,BorodinFerrari}. After that interchange of the sum and the limit is allowed. However we note that the limiting expression for the kernel still depends on which site $b^{\mathcal{B}^2}(\tau_2)$ is  via the value of $\Upsilon_{2} (\pi_2)$, which in turn depends on $\pi^{\mathcal{B}^2}(b(\tau_2))$.
To go from  the sums \eqref{fredholmsum} to integrals we note that within every summation in an index $\tau$ running through a boundary $\mathcal{B}$ the function $\Upsilon_{\mathcal{B}} (\pi^{\mathcal{B}}(b(\tau)))$ will enter linearly as a coefficient.
It turns out that  this coefficient amounts exactly to unit. This happens because the boundaries defined in  (\ref{b_1^k}) are locally straight and
the $\tau$-dependent coefficient is averaged out on a smaller scale than the fluctuational one, which affects the resulting integral.
The following lemma shows how the averaging works. After we apply it the statement of the theorem 2.2 follows.

\begin{lemma}Suppose that $b(\tau)\in \mathcal{B}$, $\tau = \tau_0+[s L^{1/3}]$,  $f_{\lim}(s)$ is a differentiable function and the following limit
\begin{equation}\label{limit}
    \lim_{L\to \infty} L^{1/3}f(b(\tau))=\Upsilon_\mathcal{B}(\pi)  f_{\lim}(s)
\end{equation}
holds  uniformly in bounded sets $s\in [a,b]$.
Then, if the boundary is close to a continuous differentiable path  in a sense (\ref{b_1^k},\ref{b_n^k}), we have
\begin{equation}\label{sum limit}
  \lim_{L\to \infty} \sum_{\tau=\tau_0+[a L^{1/3}]}^{\tau_0+[b L^{1/3}]}f(b(\tau))=\int_a^b f_{\lim} (s)ds
\end{equation}
\end{lemma}
\begin{proof} The proof is based on the fact that the order of the correction term accounting for the difference between
the boundary on the lattice and its continuous differentiable counterpart allows one to consider the boundary as locally straight at the scales up  to the fluctuation scale. This in particular means that in such a  small  scale, where the site-independent part of the limiting function can be considered as constant, the site-dependent part can be summed separately. It turns out that under this summation the site dependence exactly cancels with the slope dependence defined at the macroscopic scale, so that the remaining expression converges to integral of the site-independent part only.

To be specific, let us divide the range of summation into bins of size $\varepsilon L^{1/3}$, where $\varepsilon$ is small, and perform the summation
in two stages: first within each bin and second over all the bins. The first summation yields
\begin{equation}
   \sum_{\tau=\tau_0+[\varepsilon n L^{1/3}]}^{\tau_0+[\varepsilon(n+1) L^{1/3}] }f(b(\tau))\simeq L^{-1/3}(N_{v}^{\varepsilon} \Upsilon_{\mathcal{B}}(p)+N_h^{\varepsilon} \Upsilon_{\mathcal{B}}(1))(f_{\lim}(\varepsilon n)+O(\varepsilon)),
\end{equation}
where  and $N_h^{\varepsilon}$ and $N_v^{\varepsilon}$ are the numbers of horizontal and vertical segments of the boundary within the summation range.
Note that the fraction of these numbers, which corresponds to the slope of the boundary, being defined on the macroscopic scale persists up to the fluctuation scale, i.e. depends only on the value of $\chi=\lim_{L\to \infty}\tau_0/L$:
\begin{eqnarray}
    N_v^{\varepsilon}&=&\delta t \simeq\varepsilon L^{1/3} \left(\frac{\partial \zeta(\theta,\chi)/ \partial \chi+1}{2}\right),\\
      N_h^{\varepsilon}&=&-\delta x \simeq -\varepsilon L^{1/3} \left(\frac{\partial \zeta(\theta,\chi)/ \partial \chi-1}{2}\right).
\end{eqnarray}
From the explicit form of $\Upsilon_{\mathcal{B}}(\pi)$, \eqref{upsilon}, we have
\begin{equation}
   \frac{\partial \zeta(\theta,\chi)/ \partial \chi+1}{2}\Upsilon_{\mathcal{B}}(p)-\frac{\partial \zeta(\theta,\chi)/ \partial \chi-1}{2}\Upsilon_{\mathcal{B}}(1)=1,
\end{equation}
i.e. $ (N_{v}^{\varepsilon} \Upsilon(p)+N_h^{\varepsilon}\Upsilon(1))\simeq\varepsilon L^{1/3}$. Finally, after taking  limit $L\to\infty$, performing the second summation $\sum_{1 \leq n \leq [(b-a)/\epsilon]}\varepsilon f_{\lim}(\varepsilon n)$ and taking limit $\epsilon \to 0$ we arrive at the desired result \eqref{sum limit}.
\end{proof}

\begin{acknowledgements}
We are grateful to Patrik Ferrari and Alexei Borodin for illuminating discussions on the terminology
of space-like and time-like paths. This work was supported by the RFBR grant 12-01-00242-a, the grant of the Heisenberg-Landau program and the DFG grant RI 317/16-1. The work of V.P. was jointly supported by RFBR grant (N 12-02-91333) and grant of NRU HSE Scientific Fund (N 12-09-0051).

\end{acknowledgements}

\end{document}